\documentclass{article}

\usepackage{microtype}
\usepackage{graphicx}
\usepackage{subcaption}
\usepackage{booktabs} 

\usepackage{amssymb}
\usepackage{pifont}
\usepackage{multirow}
\usepackage{graphicx}
\usepackage{float}
\usepackage[shortlabels]{enumitem}
\usepackage{colortbl}
\usepackage{makecell}
\usepackage{caption}
\usepackage{listings}
\usepackage{xcolor}
\usepackage{wrapfig}
\usepackage{flushend}
\usepackage{tcolorbox}
\usepackage{algpseudocode}
\usepackage{amsmath,amssymb}
\usepackage{xcolor} 
\definecolor{myrow}{HTML}{E6F5FF} 

\usepackage{hyperref}

\usepackage[preprint]{icml2026}

\usepackage{amsmath}
\usepackage{amssymb}
\usepackage{mathtools}
\usepackage{amsthm}

\usepackage[capitalize,noabbrev]{cleveref}

\theoremstyle{plain}
\newtheorem{theorem}{Theorem}[section]

\newtheorem{lemma}[theorem]{Lemma}

\theoremstyle{definition}

\newtheorem{assumption}[theorem]{Assumption}
\theoremstyle{remark}
\newtheorem{remark}[theorem]{Remark}

\usepackage[textsize=tiny]{todonotes}

\icmltitlerunning{Boosting Multimodal Federated Learning via Chained Modality Optimization}

\begin{document}

\twocolumn[
  \icmltitle{Boosting Multimodal Federated Learning via Chained Modality Optimization}



  \icmlsetsymbol{equal}{*}

  \begin{icmlauthorlist}
    \icmlauthor{Zixin Zhang}{imu}
    \icmlauthor{Fan Qi}{imu,tjut}
    \icmlauthor{Shuai Li}{tjut}
    \icmlauthor{Xiaoshan Yang}{nlpr}
    \icmlauthor{Changsheng Xu}{nlpr}
  \end{icmlauthorlist}
  \icmlaffiliation{tjut}{School of Computer Science and Engineering, Tianjin University of Technology, Tianjin, China}
  \icmlaffiliation{nlpr}{Institute of Automation, Chinese Academy of Sciences, Beijing, China}
  \icmlaffiliation{imu}{College of Computer Science, Inner Mongolia University, Hohhot, Inner Mongolia, China}

  \icmlcorrespondingauthor{Fan Qi}{fanqi@email.tjut.edu.cn}




  \vskip 0.3in
]



\printAffiliationsAndNotice{}  
\begin{abstract}
Multimodal Federated Learning (MMFL) enables privacy-preserving collaborative learning across decentralized clients with heterogeneous data and modality availability.
However, most existing MMFL methods cast multimodal training as a joint optimization problem, overlooking a key bottleneck: modality competition, where dominant modalities suppress weaker ones and lead to suboptimal global models.
To address this, we propose \textsc{FedMChain}, a balanced MMFL framework that structures federated multimodal training as a chain of modality-wise phases.
This phase-wise design gives each modality a dedicated local optimization window on multimodal clients to mitigate modality competition, and further promotes cross-modal complementarity via an error-compensated regularizer.
On the server side, we employ a sparse sign-guided aggregation strategy that leverages directional sign agreement for robust intra-modality aggregation, avoids destructive averaging, and supports less frequent synchronization to reduce communication overhead.
Extensive experiments on multimodal benchmarks demonstrate that \textsc{FedMChain} consistently improves predictive performance while requiring less frequent communication than baselines.
\end{abstract}


\section{Introduction}
\label{sec:intro}


Multimodal Federated Learning (MMFL) serves as a privacy-preserving paradigm for collaborative training over distributed multimodal data silos, holding substantial value for domains such as autonomous driving~\cite{zheng2023autofed} and intelligent healthcare~\cite{orzikulova2024federated}.
Existing MMFL works primarily tackle optimization difficulties arising from \emph{statistical heterogeneity} and \emph{modality heterogeneity} by improving local cross-modal representation alignment~\cite{bao2023multimodal,yu2023multimodal}, designing personalized aggregation and adaptive optimization to accommodate client-specific modality availability and distributions~\cite{chen2024fedmbridge,yang2024cross,gao2025multimodal,Pokharel_2025_CVPR}, or jointly optimizing alignment and aggregation for more coherent multimodal collaboration~\cite{qi2024adaptive,Phung_2025_ICCV}. However, most existing methods implicitly assume that \emph{different modalities follow a relatively balanced optimization process, enabling consistent improvements under joint training}.

\begin{figure}[t]
\centering
\includegraphics[width=0.46\textwidth]{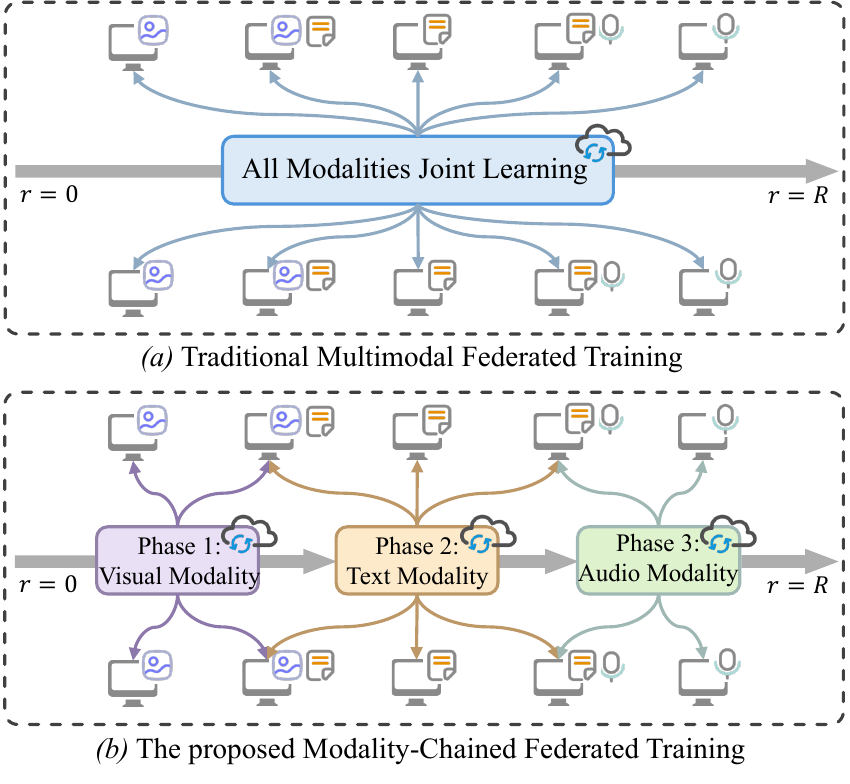}
\caption{Comparison of the training pipelines of conventional MMFL and our proposed method. Here, $r$ indexes the server communication rounds.}
\label{fig:intro}
\end{figure}

The above assumption is often violated in practice by \emph{\textbf{modality competition}}~\cite{huang2022modalitycompetitionmakesjoint,peng2022balancedmultimodallearningonthefly,du2023unimodalfeaturelearningsupervised}, a phenomenon empirically observed in centralized multimodal learning. 
During joint multimodal training, this competition biases optimization toward dominant modalities, resulting in insufficient learning of weaker ones.
Through extensive experiments, we observe that on multimodal clients in MMFL, weaker modalities often suffer from loss stagnation, while the same modalities converge normally on unimodal clients, indicating a modality dominance effect.
More importantly, the dominant modality varies across clients, driven by differences in data distribution, data quality, and model availability.
%
%
At the \emph{global} level, modality competition can further undermine aggregation stability.
%
Given that the modality competition outcome differs across clients, local models of the same modality may exhibit markedly different convergence progress. Compared with the commonly discussed conflicting update directions in federated learning (FL), this convergence-progress gap poses a more severe challenge to effective cross-client knowledge integration.
BMSFed~\cite{fan2024overcomemodalbiasmultimodal} is the closest work to our motivation.
It mitigates modality competition via modality selection, aggregating only selected client–modality branches in each round. However, this may discard potentially useful information from unselected branches.

To bridge the above gap, we propose \textsc{FedMChain}, a unified framework for balanced multimodal collaboration in heterogeneous MMFL.
As illustrated in Figure~\ref{fig:intro}, unlike conventional MMFL, which jointly optimizes all modalities in each communication round, we propose \textbf{Modality-Chained Federated Training (MCFT)}, a new MMFL training paradigm. 
MCFT structures global multimodal optimization in a chain-like manner, alternately training modalities by optimizing one modality-specific local model at a time.
Accordingly, MCFT splits MMFL into modality-wise phases, optimizing each modality in a unimodal-like regime with an exclusive optimization window to mitigate client-specific local \emph{\textbf{modality competition}} without increasing overall training time.
To promote cross-modal complementarity across phases, we design a local error-compensated regularizer that up-weights samples misclassified by preceding modalities.
Moreover, to improve the stability of intra-modality aggregation under heterogeneity, we propose \textbf{Sparse Sign-guided Consensus Aggregation (SSCA)}. 
SSCA sparsifies client updates and clusters clients by leveraging directional sign agreement, making clustering less sensitive to large cross-client differences in update magnitude. It then fuses information across clusters only on coordinates with consistent directional consensus, avoiding destructive averaging under direction conflicts. 
As a side benefit, SSCA remains stable under longer synchronization intervals, enabling a lower communication frequency. In summary, our contributions are threefold:
\begin{itemize}
    \item We propose \textsc{FedMChain}, a novel MMFL framework that performs sequential modality optimization in heterogeneous federated settings, enabling more sufficient modality-specific learning and alleviating modality competition.
    \item We introduce a sparse sign-guided aggregation strategy that leverages directional sign agreement for robust intra-modality integration, with reduced communication overhead.
    \item Extensive experiments on three widely used multimodal benchmarks demonstrate the superior effectiveness and efficiency of the proposed method.
\end{itemize}

\section{Related Works}
\label{sec:related}

\subsection{Multimodal Federated Learning}
Based on whether clients share a consistent modality configuration, MMFL can be categorized as consistent and inconsistent MMFL~\cite{che2023multimodal}. We focus on the inconsistent setting, where modality availability varies across clients.
Existing approaches for inconsistent MMFL mainly fall into two lines:
1) \textit{Modularity-based training and aggregation}. These approaches treat a part of the client models as shared modules for federated aggregation, thereby facilitating knowledge sharing among clients~\cite{yang2022cross,zhang2023unimodal,yuan2024fedmfsfederatedmultimodalfusion,cho2022heterogeneous,qi2024adaptive,li2024cross,chen2024fedmbridge,Phung_2025_ICCV,Pokharel_2025_CVPR}.
%
2) \textit{Representation-based aggregation}. Such approaches utilize prototypes or intermediate representations as carriers of knowledge, aggregating them on the server to guide the local training process~\cite{yu2023multimodal,zeng2024open,le2024cross,guo2024contribution,gao2025multimodal,seo2025clientsequalcollaborativemodel}.
%
However, most inconsistent MMFL methods still optimize modalities jointly at each client and thus do not explicitly address modality competition.
Notably, BMSFed~\cite{fan2024overcomemodalbiasmultimodal} targets modal bias by selecting modality-specific networks for communication and aggregating global prototypes to strengthen weaker modalities; yet its selection mechanism may underutilize complementary information and does not directly resolve inter-modality gradient conflicts during local joint optimization.

\begin{figure*}[ht]
\centering
\includegraphics[width=1\textwidth]{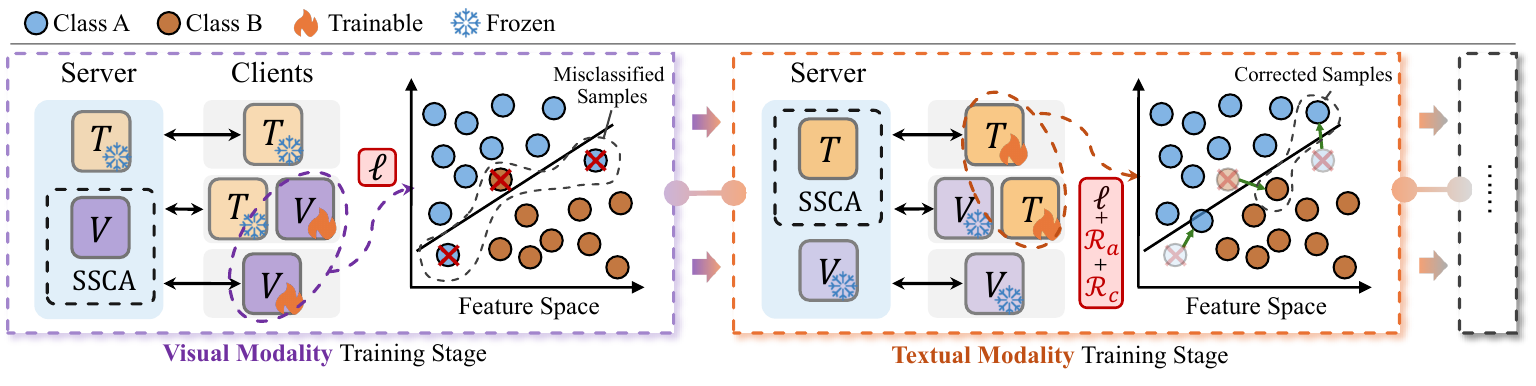}
\caption{Illustration of Modality-Chained Federated Training (MCFT), using visual ($V$) and textual ($T$) modalities as an example.}
\label{fig: training}
\end{figure*}

\subsection{Imbalanced Multimodal Learning} 
Multimodal learning often suffers from modality competition~\cite{wang2020makes,peng2022balanced}, where the dominant modality converges faster and suppresses the learning of weaker ones. 
Recent approaches address this issue from multiple angles.
Some methods~\cite{wang2020makes,peng2022balanced,fan2023pmr,li2023boosting,wei2024mmpareto} adjust the gradient magnitudes to slow down the learning of dominant modalities, thereby maintaining a more balanced optimization process across modalities.
Others~\cite{wu2022characterizing,du2023uni,zhang2023constrained} introduce additional auxiliary modules to explicitly reweight or recalibrate the contribution of each modality.
Another line of work explores improving training paradigms~\cite{fan2024detached,hua2024reconboost,jiang2024multimodal,zhang2024multimodal} or leveraging data augmentation~\cite{hwang2025midasmisalignmentbaseddataaugmentation,ma2025improving} to achieve more coordinated multimodal representation learning.
However, directly applying these methods to MMFL is challenging: gradient modulation and auxiliary recalibration add computation overhead and rely on sufficiently rich local data, while alternating training can markedly prolong local training, hurting system efficiency in communication-limited federated settings. Consequently, their gains can be limited and less stable under resource constraints and heterogeneous data.

\section{Method}
\label{sec: method}



In this section, we propose the \textsc{FedMChain} framework, comprising Modality-Chained Federated Training (MCFT) and Sparse Sign-guided Consensus Aggregation (SSCA), as shown in Figure~\ref{fig: training} and Figure~\ref{fig: merging}, respectively. Further details on \textsc{FedMChain} will be discussed in the following.

\subsection{Problem Formulation}
We consider a MMFL framework consisting of a central server and a set of clients $\mathcal{C} = \{1, 2, \dots, |\mathcal{C}|\}$.
Each client \( i \in \mathcal{C} \) holds a local dataset $\mathcal{D}_i = \big\{ (\{x_{i,j}^{(m)}\}_{m \in \mathcal{M}_i}, y_{i,j}) \big\}_{j=1}^{n_i}$,
where $n_i = |\mathcal{D}_i|$ is the number of local samples and $N = \sum_{i=1}^{|\mathcal{C}|} n_i$ is the total number of samples.
Here, $\mathcal{M}_i \subseteq \mathcal{M}$ denotes the subset of modalities available to client $i$, and $\mathcal{M}$ is the global modality set.
Each $x_{i,j}^{(m)}$ is the input of modality $m$ for sample $j$ on client $i$, and $y_{i,j}$ is the corresponding label.
Clients differ in their modality configurations: some are unimodal ($|\mathcal{M}_i| = 1$), while others are multimodal ($|\mathcal{M}_i| > 1$).
Each client $i$ maintains modality-specific encoders and classifiers with parameters $\Theta_i = \{ (\theta_{E,i}^{(m)}, \theta_{C,i}^{(m)}) \mid m \in \mathcal{M}_i \}$.
Given input \( x_{i,j}^{(m)} \), the modality prediction score is $z_{i,j}^{(m)} = f_{C,i}^{(m)}( f_{E,i}^{(m)}(x_{i,j}^{(m)}))$.
Multimodal clients may fuse predictions from multiple modalities using a fusion operator $\mathcal{F}_i(\cdot)$, 
while unimodal clients rely on a single branch.
The central server aggregates local updates to learn $K$ sets of global parameters $\{\Theta_G^k\}_{k=1}^K$, where $\Theta_G^k = \{ (\theta_{E,G}^{(m,k)}, \theta_{C,E}^{(m,k)}) \}_{m \in \mathcal{M}}$, aiming to achieve robustness under both data heterogeneity and modality heterogeneity.
At each communication round, each client $i$ is associated with exactly one global model,
represented by an assignment indicator $a_{i,k} \in \{0,1\}$ satisfying
$\sum_{k=1}^K a_{i,k} = 1$.
The global optimization objective is defined as:
\begin{equation}
\small
F_i(\Theta_G^k;\mathcal D_i)\triangleq
\frac{1}{n_i}\sum_{j=1}^{n_i}
\ell\!\left(
\mathcal{F}_i\!\left(\{z_{i,j}^{(m,k)}\}_{m\in \mathcal{M}_i}\right),
y_{i,j}
\right),
\end{equation}
\begin{equation}
\small
\min_{\{\Theta_G^k\}_{k=1}^{K}}
\sum_{i=1}^{|\mathcal{C}|} w_i
\sum_{k=1}^{K} a_{i,k}\,
F_i(\Theta_G^k;\mathcal D_i),
\end{equation}
where $w_i=\frac{n_i}{N}$, $F_i(\cdot;\cdot)$ denotes the empirical risk of client $i$, and $\ell(\cdot, \cdot)$ denotes the per-sample local loss function.

\subsection{Modality-Chained Federated Training (MCFT)}

%
Prior work~\cite{wang2020makes,peng2022balanced,fan2023pmr,du2023uni} attributes modality competition to gradient conflict across modalities under joint training, i.e., simultaneously optimizing multiple modality-specific branches under a shared objective.
Let $z_{i,j}^{\mathrm{f}} = \mathcal{F}_i(\{z_{i,j}^{(m)}\}_{m \in \mathcal{M}_i})$ denote the fused prediction scores for the $j$-th sample at client $i$.
Following the analysis proposed by Hua et al. \cite{hua2024reconboost}, the learning behavior can be characterized through a gradient alignment analysis: a modality achieves effective optimization and becomes dominant when the gradient of its modality-specific prediction is directionally consistent with that of the fused prediction,
\begin{equation}
    \Big\langle \nabla_{z_{i,j}^{(m)}} \ell,\;
    \nabla_{z_{i,j}^{\mathrm{f}}} \ell \Big\rangle > 0,
\end{equation}
In the opposite case, a persistent misalignment ($<0$) causes the modality to converge to suboptimal solutions, 
hindering its learning progress. 

\begin{figure*}[ht]
\centering
\includegraphics[width=1\textwidth]{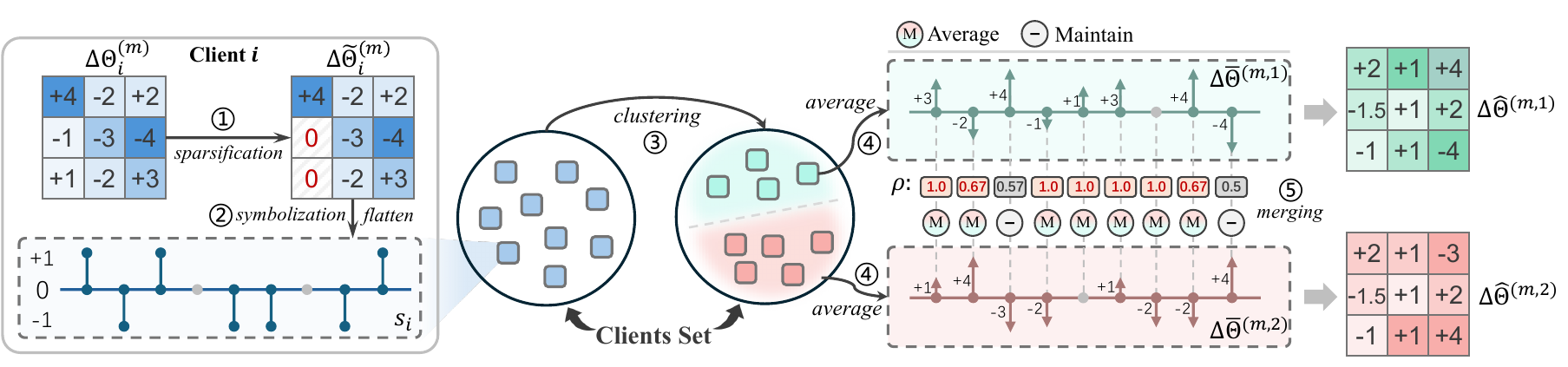}
\caption{Illustration of Sparse Sign-guided Consensus Aggregation (SSCA). We use $K=2$ as an example.}
\label{fig: merging}
\end{figure*}
In MMFL, modality competition not only undermines client-side optimization, but also carries over to server aggregation, resulting in modality-biased global updates and unstable cross-client collaboration.
To address this issue, we propose Modality-Chained Federated Training (MCFT), a \emph{modality-by-modality} optimization paradigm that updates one modality branch at a time, thereby reducing inter-modality gradient interference.
Concretely, MCFT decomposes global training into a sequence of modality-specific phases.
In phase $m$, only parameters associated with modality $m$ are optimized:
\begin{enumerate}[label=\alph*), leftmargin=*]
    \setlength{\itemsep}{5pt}
    \setlength{\parsep}{0pt}
    \setlength{\parskip}{0pt}
    \item \textbf{Unimodal clients} with modality $m$ update their modality-$m$ branch locally;
    \item \textbf{Multimodal clients} update only modality-$m$ parameters while freezing all other modality branches;
    \item The \textbf{server} periodically aggregates updates restricted to modality $m$ from participating clients to maintain modality-consistent global optimization.
\end{enumerate}
Nevertheless, phase-wise modality-specific optimization alone is insufficient to fully leverage multimodal collaboration. 
Beyond avoiding gradient interference, effective MMFL requires (i) \emph{cross-modal alignment} to induce consistent semantics across modalities and (ii) \emph{cross-modal complementarity} to ensure that distinct modalities contribute non-redundant information. 
To this end, we augment the client objective in each phase with explicit regularizers targeting both properties.
Formally, during the phase with active modality $m$, client $i$ minimizes the following mini-batch objective:
\begin{equation}
\small
\begin{aligned}
\mathcal{L}^{(m)}_i
=
&\mathbb{E}_{\mathcal{B}_i \sim \mathcal{D}_i}
\Bigg[
\frac{1}{|\mathcal{B}_i|}\sum_{(x_{i,b},y_{i,b})\in\mathcal{B}_i}
\underbrace{\ell\!\left(z_{i,b}^{(m)},\,y_{i,b}\right)}_{\text{Task}} \\
&+
\mathbb{I}\!\left[\,|\mathcal{P}_m\cap\mathcal{M}_i|>0\,\right]
\Big(
\underbrace{\lambda_a\,\mathcal{R}_{a,i}^{(m)}(\mathcal{B}_i)}_{\text{Alignment}}
+
\underbrace{\lambda_c\,\mathcal{R}_{c,i,b}^{(m)}}_{\text{Complementarity}}
\Big)
\Bigg],
\end{aligned}
\label{eq:4}
\end{equation}
where $\mathcal{P}_m$ denotes the set of modalities preceding $m$ in the chain, $\mathcal{R}_{a}$ and $\mathcal{R}_{c}$ are regularizers for cross-modal alignment and complementarity, respectively;
$\mathbb{I}[\cdot]$ enforces that these terms are applied only when the required preceding modalities are available at client.
\noindent\textbf{(i) Cross-modal Alignment Regularizer.}
We encourage the active modality $m$ to learn modality-invariant semantics by aligning its representation with those of modalities in $\mathcal{P}_m\cap\mathcal{M}_i$.
Let $h_{i,b}^{(m)}=f_{E,i}^{(m)}(x_{i,b}^{(m)})$ denote the latent feature of modality $m$ for the $b$-th sample in a mini-batch $\mathcal{B}_i$ of size $B=|\mathcal{B}_i|$.
We $\ell_2$-normalize features as $\tilde h=h/\|h\|_2$ and define an alignment regularizer that pulls together cross-modal representations of the same sample while pushing apart mismatched pairs within the mini-batch:
\begin{equation}
\small
\begin{aligned}
\mathcal{R}_{a,i}^{(m)}(\mathcal{B}_i)
= &
\frac{1}{|\mathcal{P}_m\cap\mathcal{M}_i|}
\sum_{m'\in\mathcal{P}_m\cap\mathcal{M}_i} \\
&\left[
-\frac{1}{B}\sum_{b=1}^{B}
\log
\frac{\exp\!\big(\mathcal{S}(\tilde h_{i,b}^{(m)},\tilde h_{i,b}^{(m')})/\tau\big)}
{\sum_{b'=1}^{B}\exp\!\big(\mathcal{S}(\tilde h_{i,b}^{(m)},\tilde h_{i,b'}^{(m')})/\tau\big)}
\right],
\end{aligned}
\label{eq:align_infonce}
\end{equation}
where $\mathcal{S}(\cdot,\cdot)$ denotes cosine similarity, $\tau$ is a temperature coefficient, and $b,b'$ index samples within the mini-batch.

\noindent\textbf{(ii) Cross-modal Complementarity Regularizer.}
To promote cross-modal knowledge complementarity, we introduce an error-compensation regularizer that increases the loss contribution of samples that are predicted with low confidence by modalities trained in preceding phases.
For client $i$ at the phase of modality $m$, we first aggregate the logits from preceding modalities available locally:
\begin{equation}
\bar{p}_{i,j}^{(<m)}(y_{i,j}\!\mid x_{i,j})
=
\mathrm{softmax}(\sum_{m'} z_{i,j}^{(m')})_{y_{i,j}},
\label{eq:prev_agg_comp}
\end{equation}
where $m'\in\mathcal{P}_m\cap\mathcal{M}_i$ and $\bar{p}_{i,j}^{(<m)}(y_{i,j}\mid x_{i,j})$ denotes the aggregated confidence on the ground-truth label.
We then define an error-compensation weight from the preceding-phase aggregated confidence:
\begin{equation}
e_{i,j}^{(m)} \;=\; 1-\bar{p}_{i,j}^{(<m)}\!\left(y_{i,j}\mid x_{i,j}\right),
\label{eq:weight_comp}
\end{equation}
and construct the complementarity regularizer by modulating the supervised loss of the active modality $m$:
\begin{equation}
\mathcal{R}_{c,i,j}^{(m)} \;=\; e_{i,j}^{(m)}\,\ell\!\left(z_{i,j}^{(m)},\,y_{i,j}\right).
\label{eq:8}
\end{equation}
\subsection{Sparse Sign-guided Consensus Aggregation (SSCA)}

Although MCFT mitigates modality competition by decoupling modality updates across stages, within-modality aggregation remains challenging. Under statistically non-IID client data, updates for the same modality can be highly inconsistent across clients. 
To address this issue, we propose Sparse Sign-guided Consensus Aggregation (SSCA), which leverages directional agreement to selectively consolidate intra-modality updates, reducing destructive averaging when client updates are inconsistent under statistical heterogeneity.

For a given modality $m$, each participating client $i\in\mathcal{C}$ computes a modality-specific client vector
$\Delta\Theta_i^{(m)}=\Theta_i^{(m)}-\Theta_G^{(m,k)}$, 
i.e., the directional offset between its local modality branch and the corresponding server model. 
We then apply a sparsification operator $\mathcal{T}(\Delta\Theta,\kappa)$ with retention ratio $\kappa$ to keep the most influential coordinates and zero out the rest, yielding a sparse signed update whose sign reflects the dominant optimization direction:
\begin{equation}
    s_i = \operatorname{sign}\!\big(\mathcal{T}(\Delta\Theta_i^{(m)}, \kappa)\big).
\end{equation}
The resulting sign vectors $\{s_i\}_{i\in\mathcal{C}}$ provide compact directional descriptors of client updates. 
To group clients with coherent update orientations, we build a clustering feature by flattening and concatenating all tensors in $s_i$ into a single vector, and perform unsupervised clustering in this sign feature space, yielding $K$ directional groups $\{\mathcal{C}_1,\dots,\mathcal{C}_K\}$. 
Sign-based clustering focuses on consistent update directions while attenuating magnitude heterogeneity, thereby reducing destructive averaging in subsequent within-group aggregation.

Within each cluster $\mathcal{C}_k$, we aggregate the member updates to form a cluster-level consensus update that captures their common optimization direction. 
Given the sparsified modality update $\Delta\tilde{\Theta}_i^{(m)}=\mathcal{T}(\Delta\Theta_i^{(m)},\kappa)$, we compute the cluster consensus by a sample-weighted mean:
\begin{equation}
\Delta\bar{\Theta}^{(m,k)}
=
\frac{\sum_{i \in \mathcal{C}_k} w_i\,\Delta\tilde{\Theta}_i^{(m)}}
     {\sum_{i \in \mathcal{C}_k} w_i}.
\label{eq:10}
\end{equation} 
We then coordinate the $K$ cluster consensuses through a parameter consistency score $\rho$, which measures cross-cluster directional agreement at each parameter coordinate and determines whether to unify the corresponding updates or preserve specific differences.
Concretely, for each coordinate $q$, let $\Delta\bar{\Theta}_{q}^{(m,k)}$ denote the $q$-th entry of $\Delta\bar{\Theta}^{(m,k)}$.
We define $\rho_q$ as the dominance ratio of the stronger direction:
\begin{equation}
\small
p_q=\sum_{k:\Delta\bar{\Theta}_{q}^{(m,k)}>0}\big|\Delta\bar{\Theta}_{q}^{(m,k)}\big|,\ \ 
n_q=\sum_{k:\Delta\bar{\Theta}_{q}^{(m,k)}<0}\big|\Delta\bar{\Theta}_{q}^{(m,k)}\big|,
\label{eq:11}
\end{equation}
\begin{equation}
\small
\rho_q=\frac{\max(p_q,n_q)}{p_q+n_q+\varepsilon}.
\label{eq:consistency_score}
\end{equation}
Here, $\varepsilon$ is a small constant for numerical stability. The score $\rho_q\in[0,1]$ quantifies how strongly one update direction dominates across clusters at coordinate $q$: values close to $1$ indicate that most cluster-level consensuses agree on the same sign, while smaller values imply severe sign conflicts and thus low cross-cluster consistency.

Based on $\rho_q$, SSCA applies a thresholding rule with hyper-parameter $\pi\in(0,1)$ to decide whether to merge cluster updates at coordinate $q$.
Specifically, we construct a coordinate-wise cross-cluster mask
\begin{equation}
\small
\mu_q = \mathbb{I}\!\left(\rho_q \ge \pi\right).
\label{eq:mask}
\end{equation}

\begin{table*}[ht]
\small
\centering
\caption{Performance comparison on CREMA-D, AVE, and CMU-MOSEI. We report overall accuracy (ACC, averaged over all clients), unimodal-client accuracies (ACC$_v$/ACC$_a$/ACC$_t$, averaged over clients that only have the corresponding modality), and multimodal-client accuracy (ACC$_m$, averaged over multimodal clients). Best results are in \textbf{bold} and second best are \underline{underlined}.}
\setlength{\tabcolsep}{5pt}
\begin{tabular}{l|cccc|cccc|ccccc}
\toprule
\multicolumn{1}{c|}{\multirow{2}{*}{Method}} & \multicolumn{4}{c|}{CREMA-D}     & \multicolumn{4}{c|}{AVE}         & \multicolumn{5}{c}{CMU-MOSEI}               \\ \cmidrule{2-14} 
\multicolumn{1}{c|}{}                 & ACC   & ACC$_v$ & ACC$_a$ & ACC$_m$ & ACC   & ACC$_v$ & ACC$_a$ & ACC$_m$ & ACC   & ACC$_v$ & ACC$_a$ & ACC$_t$ & ACC$_m$ \\ \midrule
Local                                 & 48.96 & 29.64  & 50.45  & 50.69    & 45.14 & 26.37  & 43.81  & 49.25    & 55.48 & 40.36  & 39.15  & 59.84  & 54.15  \\
FedAvg                                & 47.11 & 31.33  & 49.81  & 51.75    & 47.18 & 29.49  & 44.06  & 50.24    & 55.63 & 41.22  & 39.84  & 58.15  & 55.39  \\
FedProx                               & 48.35 & 30.74  & 50.69  & 51.97    & 48.25 & 30.26  & 44.68  & 51.42    & 56.59 & 40.03  & 40.24  & 60.11  & 55.45  \\
SCAFFOLD                              & 48.27 & 31.41  & 51.25  & 50.85    & 51.09 & 30.57  & 48.23  & 53.75    & 58.83 & 41.62  & 41.47  & 62.06  & 57.31  \\ \midrule
FedMSplit                             & 51.62 & 33.59  & 51.02  & 53.46    & 55.16 & 29.64  & 50.18  & 58.83    & 61.17 & 43.74  & 44.12  & 64.52  & 60.92  \\
CreamFL                               & 53.66 & 35.68  & 52.95  & 55.13    & 55.44 & 32.25  & 53.46  & 58.98    & -     & -      & -      & -      & -      \\
HAMFL                                 & 52.53 & 41.16  & 52.21  & 54.60    & 56.17 & 33.51  & 53.78  & 59.45    & 62.40 & 43.89  & 46.08  & 65.97  & 63.83  \\
M$^3$Fed                              & 53.94 & 40.15  & 51.52  & 55.98    & 55.28 & 33.05  & 54.12  & 59.36    & 61.53 & 45.14  & 48.26  & \textbf{68.46}  & 63.77  \\
FedMBridge                            & 54.88 & 39.82  & 51.95  & 56.75    & 57.33 & 35.69  & 56.33  & \underline{61.84}    & 61.29 & \underline{45.96}  & \textbf{50.38}  & 68.21  & \underline{64.68}  \\
BMSFed                                & \underline{56.14} & \textbf{44.25}  & \underline{53.65}  & \underline{58.69}    & \underline{57.62} & \underline{37.96}  & \textbf{58.25}  & 60.22    & \underline{63.05} & \textbf{46.24}  & 49.76  & 68.12  & 64.30  \\ \midrule
\rowcolor{orange!10}Ours              & \textbf{58.36} & \underline{43.89}  & \textbf{54.10}  & \textbf{60.68}    & \textbf{59.85} & \textbf{38.49}  & \underline{57.95}  & \textbf{62.74}    & \textbf{64.96} & 45.86  & \underline{50.25}  & \underline{68.37}  & \textbf{65.75}  \\ \bottomrule
\end{tabular}
\label{Table:big}
\end{table*}

For $\mu_q=1$, we first identify the dominant direction
\begin{equation}
\small
d_q=\operatorname{sign}(p_q-n_q),
\label{eq:14}
\end{equation}
and aggregate only the cluster entries aligned with $d_q$:
\begin{equation}
\small
\Delta\hat{\Theta}_q^{(m)}
=
\frac{\sum_{k=1}^{K}\alpha^k\,\Delta\bar{\Theta}_{q}^{(m,k)}\,\mathbb{I}\!\left(\operatorname{sign}(\Delta\bar{\Theta}_{q}^{(m,k)})=d_q\right)}
{\sum_{k=1}^{K}\alpha^k\,\mathbb{I}\!\left(\operatorname{sign}(\Delta\bar{\Theta}_{q}^{(m,k)})=d_q\right)+\varepsilon},
\label{eq:15}
\end{equation}
where $\alpha^k=\sum_{i\in\mathcal{C}_k}w_i$ is the total weight of cluster $\mathcal{C}_k$.
For $\mu_q=0$ (i.e., $\rho_q<\sigma$), we do not enforce cross-cluster merging and simply retain $\Delta\hat{\Theta}_{q}^{(m,k)}=\Delta\bar{\Theta}_{q}^{(m,k)}$.
Combining the two cases, the final aggregated update can be expressed compactly as
\begin{equation}
\small
\Delta\hat{\Theta}_{q}^{(m,k)}
=
\mu_q\,\Delta\hat{\Theta}_q^{(m)}
+
(1-\mu_q)\,\Delta\bar{\Theta}_{q}^{(m,k)}.
\label{eq:16}
\end{equation}
This rule unifies clusters only on coordinates with clear directional dominance, while keeping sign-ambiguous coordinates at the cluster-level consensuses to avoid interference from conflicting optimization directions.

Finally, the server updates the modality-specific global parameters via a controlled merge step:
\begin{equation}
\Theta_{G,r+1}^{(m,k)}=\Theta_{G,r}^{(m,k)}+\lambda_{\mathrm{merge}}\Delta\hat{\Theta}^{(m,k)},
\label{eq:lambda_merge}
\end{equation}
where $\lambda_{\mathrm{merge}}\in(0,1]$ regulates the integration strength of the aggregated update to stabilize optimization under heterogeneous client updates. 

Notably, SSCA reconciles divergent client updates by emphasizing directional consensus during aggregation, which improves robustness when client updates are misaligned under heterogeneity. Empirically, this robustness makes SSCA less sensitive to longer synchronization intervals (Table~\ref{tab:period}), allowing the server to aggregate less frequently and thus reducing communication overhead.

\subsection{Convergence Analysis}
We establish a convergence guarantee for \textsc{FedMChain} in smooth non-convex federated optimization.
For an arbitrary modality stage $m$ and an arbitrary cluster $k$, we consider the corresponding stage objective
$F^{(m,k)}$ induced by the participating clients in $\mathcal{C}_k$.
Under standard federated assumptions, Theorem~D.13 in Appendix~\ref{app:convergence} proves the canonical non-convex stationarity bound
\begin{equation}
\frac{1}{R}\sum_{r=0}^{R-1}\mathbb{E}\bigl[\|\nabla F^{(m,k)}(\Theta_r)\|^2\bigr]
\le \mathcal{O}\!\left(\frac{1}{R}\right) + \mathcal{O}(\sigma_{\mathrm{agg}}^{2}),
\end{equation}
where $R$ is the number of server aggregation rounds and $\sigma_{\mathrm{agg}}^{2}$ captures the SSCA-induced aggregation noise.
\textbf{\textit{The details of the proof are provided in Appendix~D.}}

\section{Experiments}
\label{sec: exp}

\subsection{Experimental Setup}

\noindent\textbf{Datasets.}
We evaluate our proposed method on three commonly used multimodal benchmarks: \textbf{CREMA-D}~\cite{cao2014crema}, \textbf{AVE}~\cite{tian2018audio}, and \textbf{CMU-MOSEI}~\cite{zadeh2018multimodal}.
To simulate federated multimodal settings with heterogeneous data distributions, we partition each dataset into multiple clients using a Dirichlet distribution with concentration parameter $\beta \in\{ 0.5,1.0\}$\footnote{Unless otherwise specified, all results reported in this paper are obtained with $\beta=1.0$.},where a smaller $\beta$ indicates higher heterogeneity. 
Within each client, the local data are split into 80\% for training and 20\% for testing. Additionally, we simulate heterogeneous modality availability by masking modalities for a subset of clients to create diverse cross-client modality configurations.
\textbf{\textit{Further dataset and partitioning details are included in the Appendix~\ref{app:Additional Experimental Setup-Dataset}}}.

\noindent\textbf{Implementation Details.}
All methods are implemented in PyTorch~\cite{paszke2017automatic} and trained on 8 NVIDIA RTX 3090 GPUs.
For a fair comparison, all baselines are evaluated under the same data partition protocol and identical backbone/model configurations.
The local batch size is set to 128, and each communication round runs one local epoch.
We train for 500 communication rounds and perform server aggregation every 20 rounds.
For our method, we set the sparsity ratio $\kappa$ to 0.7, the consensus threshold $\pi$ to 0.9, the merge learning rate $\lambda_\text{merge}$ to 0.9, and the loss weighting coefficients $\lambda_{a}$ and $\lambda_{c}$ to 0.4 and 1.0, respectively.
For all datasets, we repeat each experiment three times with different random seeds, and report the mean performance.
\textbf{\textit{Further experimental setup details are in Appendix~\ref{app:Additional Experimental Setup-Train}.}}

\noindent\textbf{Baselines.}
We compare our method with representative MMFL baselines under modality distribution heterogeneity:
\textbf{Local} (where each client is trained independently without communication), classic FL methods adapted to multimodal settings (\textbf{FedAvg}~\cite{mcmahan2017communication}, \textbf{FedProx}~\cite{li2020federated}, \textbf{SCAFFOLD}~\cite{karimireddy2020scaffold}), and recent heterogeneous MMFL approaches (\textbf{HAMFL}~\cite{qi2024adaptive}, \textbf{M$^3$Fed}~\cite{li2024cross}, \textbf{CreamFL}~\cite{yu2023multimodal}, \textbf{FedMSplit}~\cite{chen2022fedmsplit}, \textbf{FedMBridge}~\cite{chen2024fedmbridge}, \textbf{BMSFed}~\cite{fan2024overcomemodalbiasmultimodal}).
Among them, \textbf{BMSFed} is specifically designed to alleviate modality competition in MMFL scenarios, serving as a particularly relevant and competitive baseline to our approach. 

%


\subsection{Performance Analysis}

Table~\ref{Table:big} summarizes the performance of our method and representative baselines on three datasets.
Overall, our approach achieves the best ACC (overall clients) and ACC$_m$ (multimodal clients) across the three datasets, suggesting that it more effectively mitigates modality competition and thus better preserves cross-modal complementarity during collaborative training. 
Although we are not the top method on every unimodal subset (e.g., ACC$_v$ on CREMA-D can be slightly lower than BMSFed), we consistently outperform classical FL baselines such as FedAvg/FedProx/SCAFFOLD on unimodal-client metrics, and often exceed Local training. Since unimodal clients cannot leverage cross-modal fusion locally, these gains primarily reflect more reliable server-side aggregation under non-IID updates, providing indirect evidence of the robustness of our aggregation strategy.
BMSFed is the most competitive baseline and can yield higher unimodal-client performance in some cases, consistent with its modality-aware coordination that explicitly targets modality competition and better protects unimodal branches from being overwhelmed by multimodal objectives.
In contrast, a group of heterogeneous-MMFL methods (e.g., FedMSplit, CreamFL, HAMFL, M$^3$Fed, and FedMBridge) achieves weaker overall results. Although these approaches introduce modularization or bridging to cope with missing modalities, they still do not sufficiently suppress inter-modality gradient conflicts, so modality competition persists during training and ultimately limits both ACC and ACC$_m$.
Finally, classical FL algorithms may even underperform Local training (e.g., FedAvg 47.11 vs. Local 48.96 on CREMA-D), indicating that naive parameter averaging can amplify conflicting updates and lead to negative transfer rather than effective knowledge sharing.

\begin{table}[t]
\caption{Ablation study of our framework on CREMA-D, AVE, and CMU-MOSEI.}
\setlength{\tabcolsep}{6pt}
\centering
\small
\begin{tabular}{c|ccc}
\toprule
Method                                      & CREMA-D & AVE   & CMU-MOSEI \\ \midrule
w/o MCFT & 
54.96\,{\textcolor{gray}{\scriptsize$\downarrow$3.40}} & 
54.82\,{\textcolor{gray}{\scriptsize$\downarrow$5.03}} & 
60.39\,{\textcolor{gray}{\scriptsize$\downarrow$4.57}} \\

w/o SSCA & 
56.25\,{\textcolor{gray}{\scriptsize$\downarrow$2.11}} & 
55.43\,{\textcolor{gray}{\scriptsize$\downarrow$4.42}} & 
60.98\,{\textcolor{gray}{\scriptsize$\downarrow$3.98}} \\ 

w/o $\mathcal{R}_a$ & 
57.12\,{\textcolor{gray}{\scriptsize$\downarrow$1.24}} & 
58.14\,{\textcolor{gray}{\scriptsize$\downarrow$1.71}} & 
64.25\,{\textcolor{gray}{\scriptsize$\downarrow$0.71}} \\

w/o $\mathcal{R}_c$ & 
57.03\,{\textcolor{gray}{\scriptsize$\downarrow$1.33}} & 
57.29\,{\textcolor{gray}{\scriptsize$\downarrow$2.56}} & 
63.84\,{\textcolor{gray}{\scriptsize$\downarrow$1.12}} \\

w/o $\mathcal{R}_a \& \mathcal{R}_c$ & 
56.31\,{\textcolor{gray}{\scriptsize$\downarrow$2.05}} & 
55.60\,{\textcolor{gray}{\scriptsize$\downarrow$4.25}} & 
63.08\,{\textcolor{gray}{\scriptsize$\downarrow$1.88}} \\ \midrule

\rowcolor{orange!10}Ours & \textbf{58.36} & \textbf{59.85} & \textbf{64.96} \\ 
\bottomrule
\end{tabular}
\label{tab:ablation}
\end{table}

\begin{figure}[t]
    \centering

    \begin{subfigure}[t]{0.236\textwidth}
        \centering
        \includegraphics[width=\linewidth]{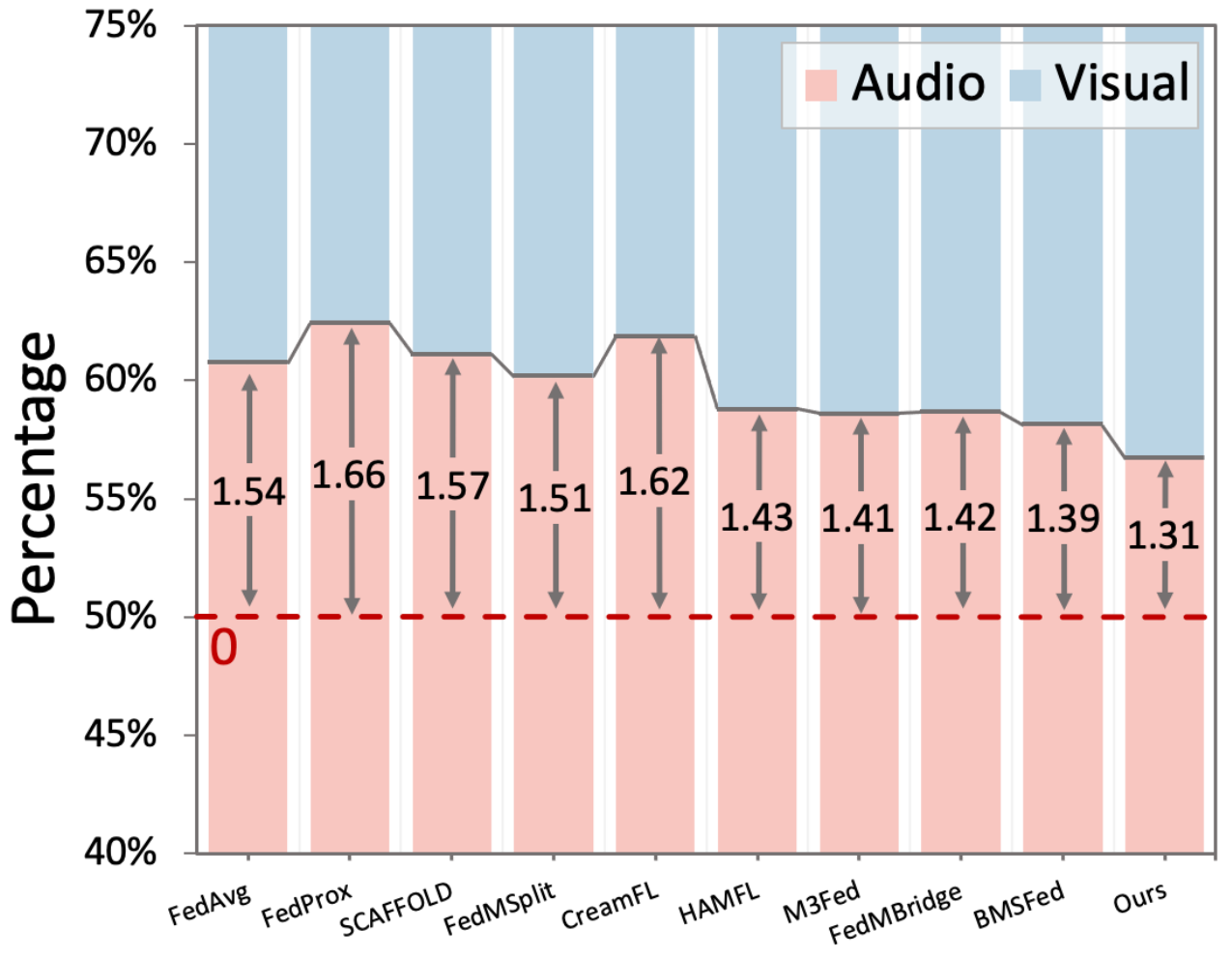}
        \caption{MIR}
        \label{fig:ex-mir}
    \end{subfigure}\hfill
    \begin{subfigure}[t]{0.23\textwidth}
        \centering
        \includegraphics[width=\linewidth]{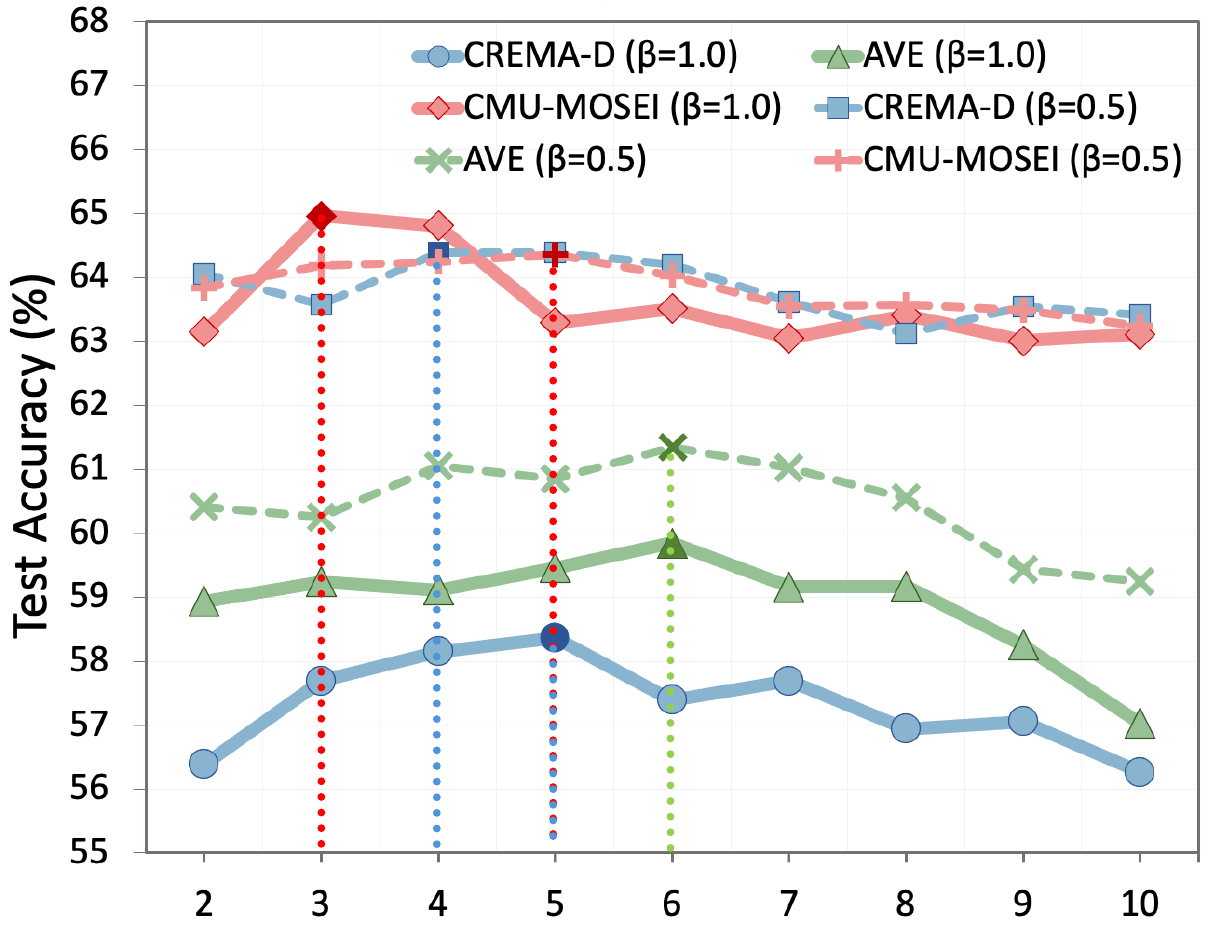}
        \caption{$K$}
        \label{fig:ex-cluster}
    \end{subfigure}

    \caption{Analysis of modality competition and clustering sensitivity. (a) Modality imbalance ratio (MIR) on CREMA-D. (b) Effect of the number of SSCA clusters $K$ on test accuracy across three datasets under different data heterogeneity levels $\beta\in\{0.5,1.0\}$.}
    \label{fig:ex-2}
\end{figure}

\begin{figure}[t]
\centering
\includegraphics[width=0.48\textwidth]{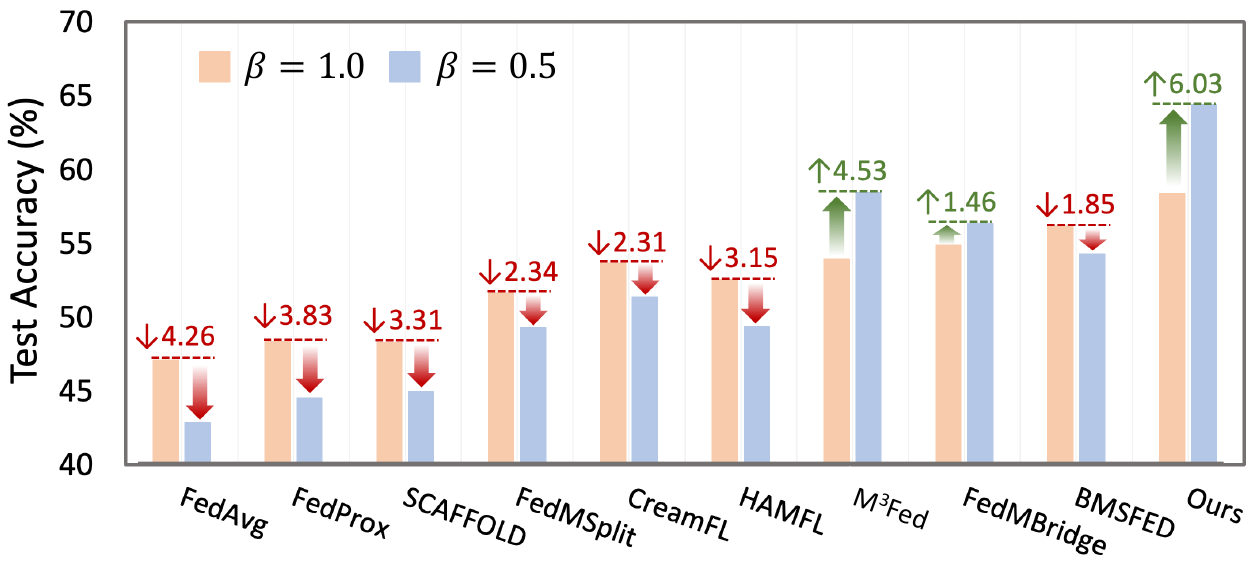}
\caption{The effect of different data heterogeneity on performance of all competitors.}
\label{fig:yizhi}
\end{figure}

\subsection{Ablation Studies}

Table~\ref{tab:ablation} isolates the effect of each component. 
Removing MCFT causes the largest performance degradation, indicating that modality-wise optimization is essential for mitigating modality competition.
Disabling SSCA also leads to a clear drop. In this setting, the server aggregation degenerates to a plain averaging scheme applied separately to each modality branch. This indicates that SSCA is necessary to stabilize collaborative optimization and to prevent naive averaging from causing unstable aggregation under heterogeneous client updates.
Removing either $\mathcal{R}_a$ or $\mathcal{R}_c$ results in a noticeable degradation, and ablating both produces an even larger drop. Specifically, $\mathcal{R}_a$ promotes cross-modal alignment, which facilitates reliable knowledge transfer among modality branches, whereas $\mathcal{R}_c$ reweights training to make each modality focus more on samples that were misclassified by preceding modalities, thereby promoting complementary evidence rather than redundant signals. Consequently, removing both terms simultaneously breaks alignment and weakens complementarity, leading to the most pronounced performance loss.

\begin{figure*}[t]
    \centering

    \begin{subfigure}[t]{0.23\textwidth}
        \centering
        \includegraphics[width=\linewidth]{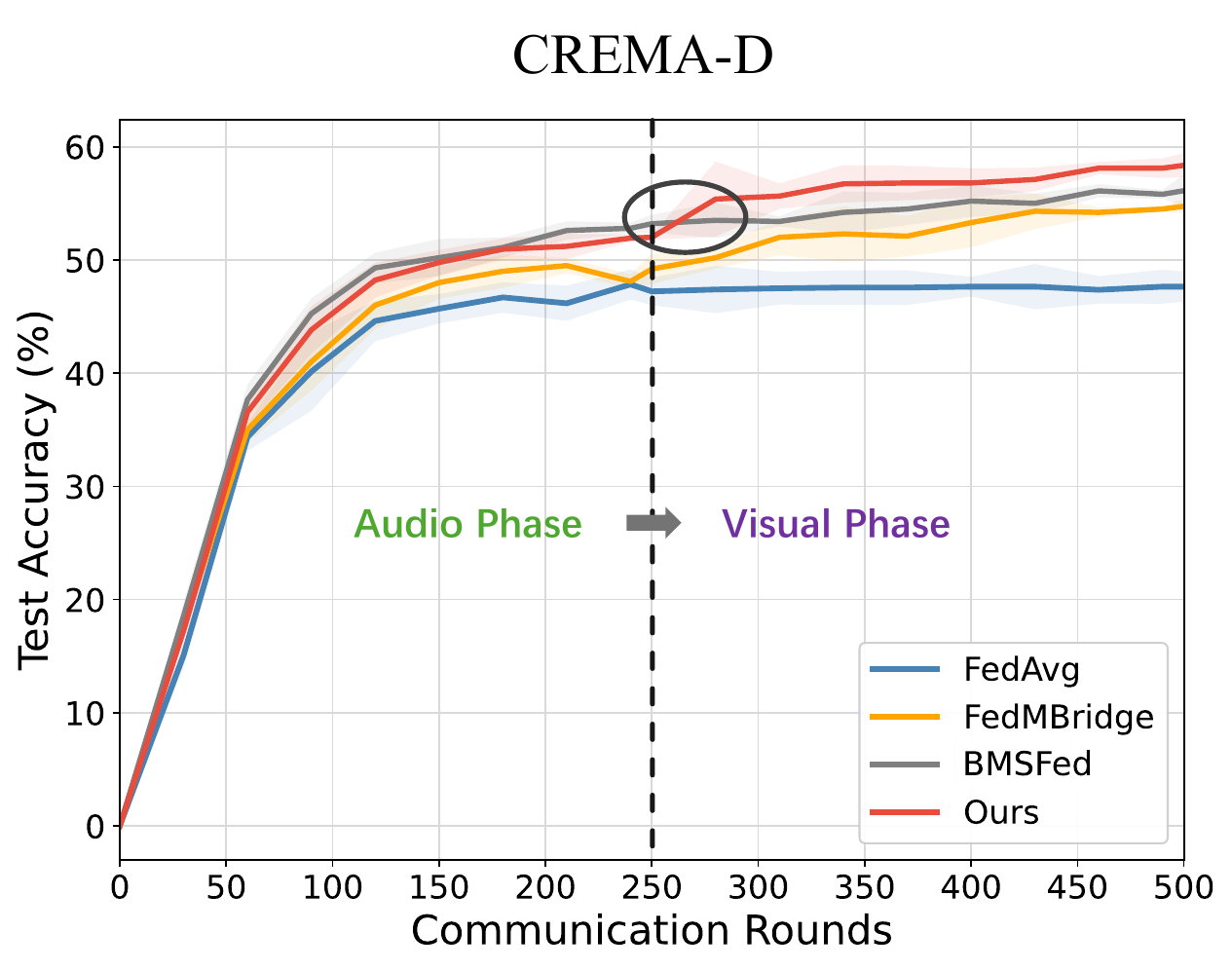}
        \caption{}
        \label{fig:round-C}
    \end{subfigure}\hfill
    \begin{subfigure}[t]{0.23\textwidth}
        \centering
        \includegraphics[width=\linewidth]{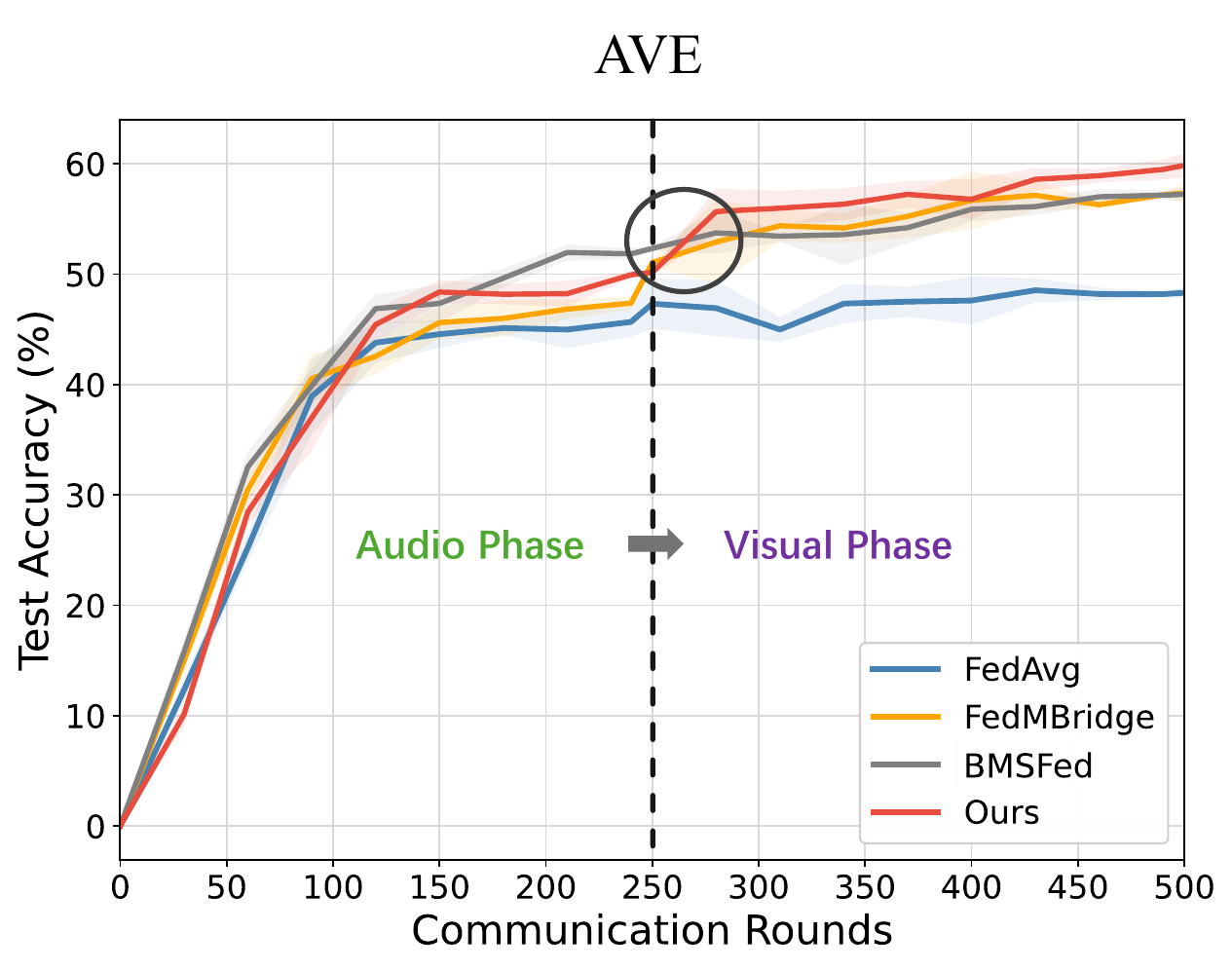}
        \caption{}
        \label{fig:round-A}
    \end{subfigure}\hfill
    \begin{subfigure}[t]{0.23\textwidth}
        \centering
        \includegraphics[width=\linewidth]{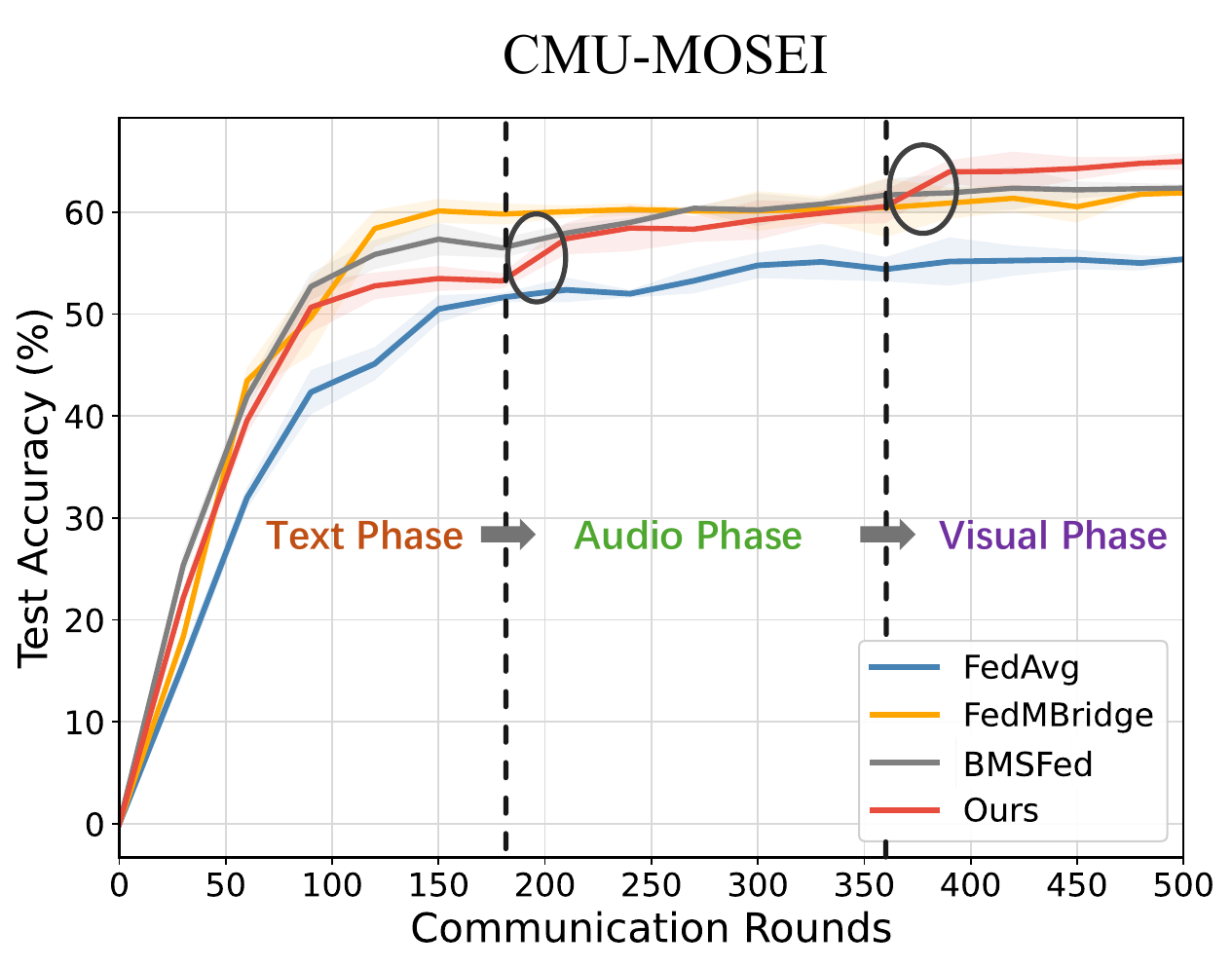}
        \caption{}
        \label{fig:round-M}
    \end{subfigure}\hfill
    \begin{subfigure}[t]{0.23\textwidth}
        \centering
        \includegraphics[width=\linewidth]{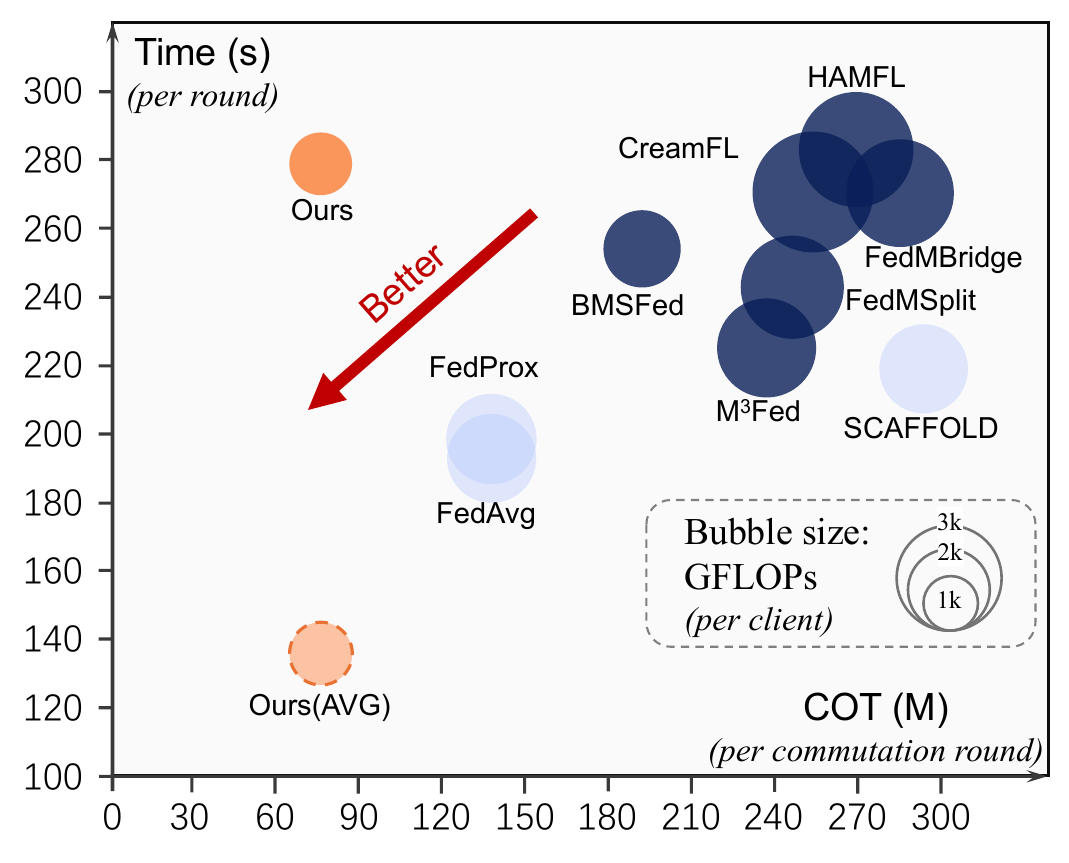}
        \caption{}
        \label{fig:eff}
    \end{subfigure}

    \caption{(a)-(c) The effect of communication rounds on the performance of the three datasets; (d) Efficiency trade-offs for all competitors on CREMA-D (x-axis: communication cost per round; y-axis: average runtime per round; bubble size: local computation per round). }
    \label{fig:round}
\end{figure*}


\subsection{Further Analysis}


\noindent\textbf{Modality competition.}
Modality competition often enlarges the performance gap between modalities.
To quantify this effect, we report the modality imbalance ratio (MIR), where a larger value indicates more severe modality imbalance (see Appendix~\ref{app:Additional Experimental Results} for the formal definition).
As shown in Figure~\ref{fig:ex-mir} on CREMA-D, our method achieves the smallest MIR among all compared approaches (1.31).
This reduced imbalance is achieved together with the best ACC/ACC$_m$ in Table~\ref{Table:big}, showing that our method improves modality balance without sacrificing the stronger modality's performance.

\noindent\textbf{Effect of the number of clusters.}
We evaluate our method with different numbers of clusters, setting $K \in \{2,3,\ldots,10\}$.
Figure~\ref{fig:ex-cluster} reports test accuracy on CREMA-D, AVE, and CMU-MOSEI under $\beta \in \{0.5,1.0\}$.
Across datasets and both $\beta$ values, performance improves as $K$ increases from $2$ to roughly $3\sim 6$, and then plateaus or mildly degrades when $K \ge 7$.
In particular, the best (or near-best) performance is typically achieved around $K=4\sim6$, suggesting that a moderate clustering granularity provides a good balance between capturing intra-class diversity and avoiding over-fragmentation.
When $K$ is too large, the clustering becomes overly fine-grained and introduces noisy assignments, which can hurt generalization.
Overall, the method is relatively insensitive within the moderate range, and we set $K=5$ as a default in all experiment.

\noindent\textbf{Degree of Heterogeneity of Data Distribution.} We investigate the performance of different methods on CREMA-D under two Dirichlet heterogeneity levels ($\beta \in \{0.5,1.0\}$). As shown in Figure~\ref{fig:yizhi}, classic FL baselines exhibit pronounced performance degradation as heterogeneity increases. In contrast, MMFL approaches such as FedMBridge improve in the more heterogeneous regime, owing to their personalized parameter aggregation or update mechanisms that mitigate cross-client negative transfer. Notably, our method achieves the best results under both settings and further improves under stronger heterogeneity (+6.03\%). 
We attribute this improvement to SSCA, it groups clients by directional sign agreement and aggregates within groups, which becomes increasingly beneficial as non-IID intensifies.

\noindent\textbf{Number of Communication Rounds.} Figure~\ref{fig:round} shows the average test performance across all clients over different communication rounds. Due to the chained modality training paradigm, our method exhibits a distinct phase-wise performance improvement pattern, with noticeable gains at each modality transition. Through successive modality updates, these phase-wise improvements accumulate, enabling our method to outperform all baselines in later communication rounds.
Notably, this advantage is achieved under the same number of communication rounds as the baselines, indicating higher performance for the same overall training duration.

\noindent\textbf{Sensitivity of the Aggregation Period.} Table~\ref{tab:period} compares different aggregation periods $\upsilon \in \{1,5,10,15,20,30\}$, where the server aggregates once every $\upsilon$ rounds.
Performance improves as $\upsilon$ increases from $1$ to $20$, and then becomes flat or slightly degrades for larger $\upsilon$.
This trend can be explained by the reliance of SSCA on directional sign consistency.
SSCA clusters clients using the sparsified signed updates and enables cross-cluster merging only on coordinates that satisfy the consensus threshold.
As $\upsilon$ increases, local optimization yields cleaner and more consistent update directions, making the sign-based grouping more reliable, which in turn improves aggregation quality.
Beyond a certain point, however, longer local trajectories under non-IID data amplify drift and reduce cross-client directional agreement, so fewer coordinates exceed the consensus threshold and the benefit of SSCA saturates or mildly declines.

\noindent\textbf{Efficiency Analysis.} We illustrate the efficiency comparison of different methods in Figure~\ref{fig:eff}.
%
Our method attains the lowest per-client local computation (smallest bubbles) and the lowest per-round communication cost (leftmost position) among all compared methods.
This is because MCFT updates only one modality branch per stage while freezing the others, reducing both backpropagation and the transmitted payload.
Regarding the per-round time, the dominant cost of our method comes from the global aggregation stage.
With SSCA enabling stable aggregation at a lower frequency, we additionally report \textbf{Ours(AVG)}, which measures the average runtime over all communication rounds (including both aggregation and non-aggregation rounds). \textbf{Ours(AVG)} achieves the lowest average time, indicating that our method has the shortest overall training time among all competitors.


\textit{\textbf{Due to space limitations, additional experimental results are provided in Appendix~\ref{app:Additional Experimental Results}.}}

\begin{table}[t]
\caption{The effect of aggregation period on the performance of the three datasets.}
\centering
\small
\setlength{\tabcolsep}{4pt}
\begin{tabular}{l|cccccc}
\toprule
\multicolumn{1}{c|}{\textit{period} ($\upsilon$)}  & 1 & 5 & 10 & 15 & 20 & 30 \\ \midrule
CREMA-D                              & 57.69 & 57.93 & 58.08 & 58.21  & \textbf{58.36}  & 58.33    \\
AVE                                  & 58.76 & 59.38 & 59.71 & 59.64  & \textbf{59.85}  & 59.42    \\
CMU-MOSEI                            & 64.02 & 64.34 & 64.19 & 64.22  & \textbf{64.96}  & 64.83    \\ \bottomrule
\end{tabular}
\label{tab:period}
\end{table}

\section{Conclusion \& Limitation}


In this work, we identify modality competition as a key bottleneck in MMFL. We propose \textsc{FedMChain}, which structures training into modality-wise phases to balance modality learning and employs a sparse sign-guided aggregation for conflict-aware, communication-efficient updates. Extensive experiments demonstrate the effectiveness of our method.
A limitation is that the modality training order is currently chosen empirically; developing more adaptive and principled ordering strategies is left for future work.

\section*{Impact Statement}

This paper presents work whose goal is to advance the field of Machine
Learning. There are many potential societal consequences of our work, none
which we feel must be specifically highlighted here.

\nocite{langley00}

\bibliography{example_paper}
\bibliographystyle{icml2026}

\newpage
\appendix
\onecolumn

\section{Appendix Overview}
Given the space limitations in the main paper, we defer additional materials to the appendix. Appendix~\ref{app:Additional Experimental Setup} provides further implementation and experimental setup details, Appendix~\ref{app:Additional Experimental Results} reports additional experimental results, Appendix~\ref{app:convergence} contains the convergence analysis and corresponding proofs, and Appendix~\ref{app:code} summarizes the proposed method with pseudocode.

\section{Additional Experimental Setup}\label{app:Additional Experimental Setup}

\subsection{Datasets}\label{app:Additional Experimental Setup-Dataset}
\paragraph{CREMA-D} CREMA-D is an audio-visual dataset for speech emotion recognition. It contains 7,442 original clips (2–3 seconds) from 91 actors speaking short utterances, annotated with six emotion classes (anger, disgust, fear, happy, neutral, and sad). Labels are obtained via crowd-sourcing from 2,443 raters.

\paragraph{AVE} The Audio-Visual Event (AVE) dataset is a benchmark for audio-visual event localization, consisting of 4,143 YouTube videos from 28 event categories. Each video is temporally annotated with audio-visual event boundaries and contains at least one event segment of 2 seconds. The events span diverse domains including human/animal activities, musical performances, and vehicle-related sounds.

\paragraph{CMU-MOSEI} CMU-MOSEI is a large-scale multimodal sentiment/emotion dataset, comprising 23,453 annotated utterances extracted from over 5,000 videos, with more than 1,000 distinct speakers and roughly 250 topics, offering substantial diversity in both content and speaker characteristics. 
Human annotators label each sample with a sentiment score from -3 (strongly negative) to +3 (strongly positive). We view this as a three classification problem, with the categories being negative, neutral, and positive.

\begin{table*}[ht]
\caption{Statistics of the datasets used in our experiments.}
\centering
\small
\setlength{\tabcolsep}{1pt}
\begin{tabular}{lcccccc}
\toprule
\multirow{2}{*}{Dataset} &
\multirow{2}{*}{\#sample} &
\multirow{2}{*}{\#class} &
\multirow{2}{*}{\#client}&
\multicolumn{3}{c}{Modality}  \\
\cmidrule(lr){5-7}
 &  &  &  & Visual & Audio & Text \\
\midrule
CREMA-D & 7,442 & 6 & 30 & \ding{51} & \ding{51} & \ding{55}  \\
AVE & 3,741 & 28 & 30 & \ding{51} & \ding{51} & \ding{55}  \\
CMU-MOSEI & 22,346 & 3 & 70 & \ding{51} & \ding{51} & \ding{51}  \\
\bottomrule
\end{tabular}
\label{tab:datasets}
\end{table*}


As shown in Figure~\ref{fig:client_label_dist_2}-Figure~\ref{fig:client_label_dist_1}, we present heatmaps of the client-wise label distributions induced by our Dirichlet-based partitioning under two heterogeneity levels ($\beta\in\{0.5,1.0\}$). For each dataset, we construct a client--class matrix and normalize each client’s label counts into proportions; the heatmap visualizes these proportions with color intensity, where columns denote clients and rows denote classes. In particular, a smaller concentration parameter ($\beta=0.5$) leads to more skewed and client-specific label compositions, whereas a larger value ($\beta=1.0$) yields comparatively more balanced label distributions across clients. As shown in Figure~\ref{fig: client_modality}, we further illustrate the client-wise modality availability resulting from our modality partitioning strategy. Specifically, we construct a binary client--modality matrix in which each entry indicates whether a modality is available at a given client; the heatmap uses color to distinguish available versus missing modalities, with columns denoting clients and rows denoting modalities. Such modality-level non-IID settings reflect realistic multimodal federated scenarios, where acquisition or system constraints lead to systematic modality missingness across participants.

\begin{figure*}[h]
\centering
\includegraphics[width=0.9\textwidth]{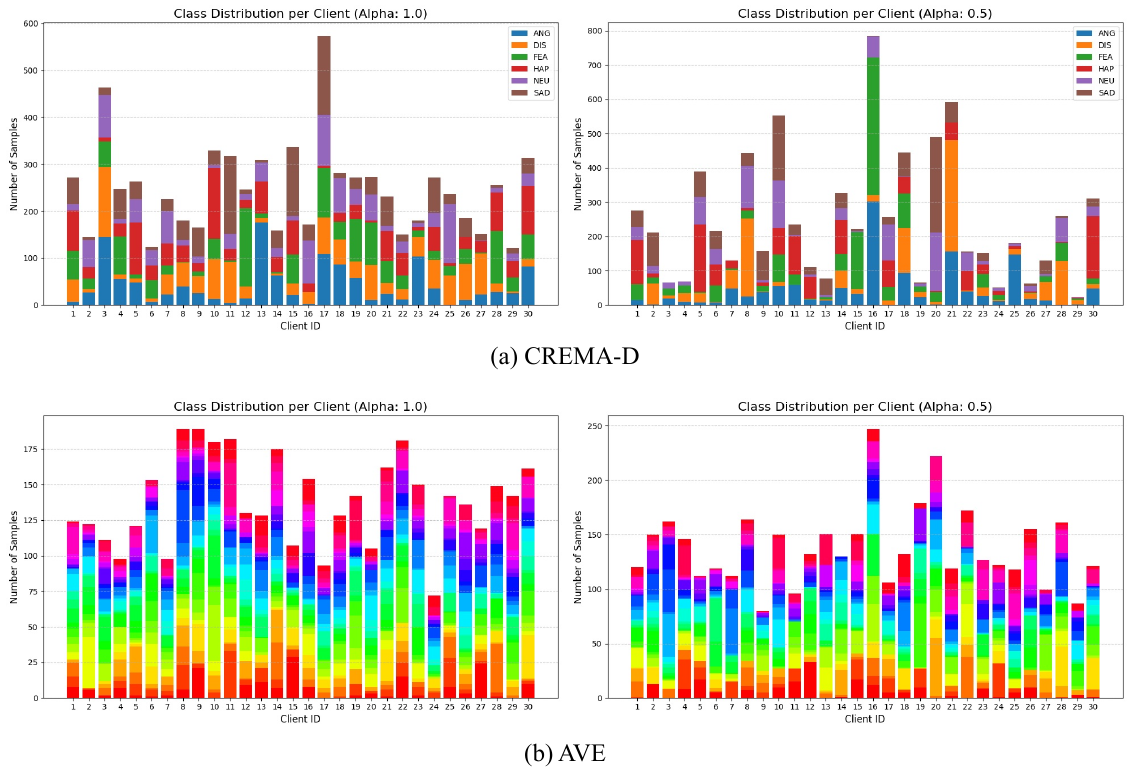}
\caption{Client-wise label distributions on CREMA-D and AVE dataset under different heterogeneity levels ($\beta \in \{0.5,1.0\}$).}
\label{fig:client_label_dist_2}
\end{figure*}

\begin{figure*}[h]
\centering
\includegraphics[width=0.9\textwidth]{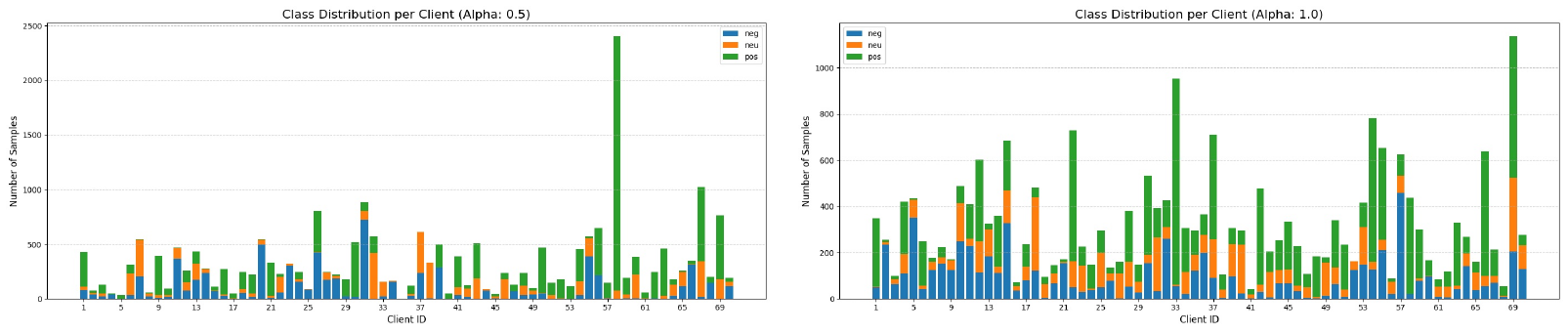}
\caption{Client-wise label distributions on CMU-MOSEI dataset under different heterogeneity levels ($\beta \in \{0.5,1.0\}$).}
\label{fig:client_label_dist_1}
\end{figure*}

\begin{figure*}[h]
\centering
\includegraphics[width=1\textwidth]{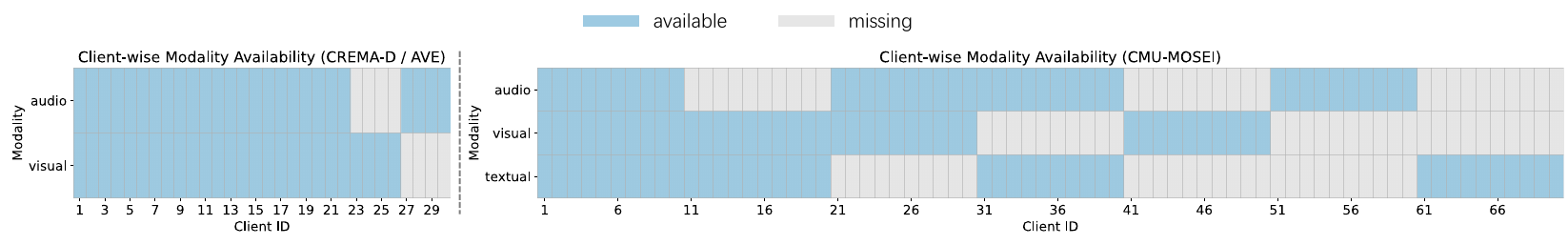}
\caption{Client-wise modality availability on CREMA-D/AVE and CMU-MOSEI.}
\label{fig: client_modality}
\end{figure*}

\subsection{Training Details}\label{app:Additional Experimental Setup-Train}
With respect to the CREMA-D and AVE datasets, we adopt a ResNet-18 backbone with modality-specific input stems (3-channel for vision and 1-channel for audio). For both datasets, we use one video frame of size $224 \times 224 \times 3$, extracted as the middle frame of each clip (with random resized cropping and horizontal flipping during training). For the audio stream, we load the waveform at $22{,}050$ Hz and ensure a fixed duration of $3$ seconds by truncation or repetition. We then compute a log-magnitude STFT spectrogram using $\text{n}\_\text{fft}=512$ and $\text{hop}\_\text{length}=353$, resulting in $257$ frequency bins (i.e., $1 + \text{n}\_\text{fft}/2$). The visual and audio backbones output $512$-dimensional global features via adaptive pooling, which are fed into a fully connected layer to produce modality-specific predictions. Finally, the audio and visual predictions are fused to obtain the final score.

For the CMU-MOSEI dataset, we use pre-extracted language, visual, and acoustic features. The text modality is represented by BERT embeddings \cite{devlin2019bertpretrainingdeepbidirectional} with a 768-dimensional feature size, while the visual and audio modalities use FACET features (35 dimensions) and COVAREP features \cite{6853739} (74 dimensions), respectively. These features are passed through modality-specific encoders to produce 128-dimensional latent representations. In particular, AudioNet and VisualNet adopt a three-layer MLP backbone, where each layer is followed by ReLU, dropout, and layer normalization, and then a linear layer is used to output modality-specific predictions. For the text modality, TextNet uses a bidirectional LSTM to encode the input sequence; the final hidden state is projected to a 128-dimensional vector and fed into a linear prediction head. Finally, we combine the predictions from different modalities to obtain the final output.

For all three datasets, we use the Adam optimizer with dataset-specific learning rates of $1\times 10^{-5}$, $1\times 10^{-4}$, and $1\times 10^{-3}$, respectively. In each communication round, we randomly sample a fraction of clients with sampling ratio $0.7$ to participate in training. We run federated training for 500 rounds in total, and allocate the same number of training rounds to each modality. For the SSCA, we perform direction-consistency clustering using the standard KMeans algorithm.

\newpage
\section{Additional Experimental Results}\label{app:Additional Experimental Results}

\begin{wraptable}{r}{0.45\textwidth}
  \centering
  \caption{The effect of MCFT modality training order on the performance of the three datasets.}
  \vspace{-0.5\baselineskip} 
  \begin{tabular}{c|c|l}
\toprule[1pt]
\multirow{2}{*}{CREMA-D}   & Audio $\rightarrow$ Visual   & \textbf{58.36} \\
                           & Visual$\rightarrow$Audio   & 56.09 \\ \midrule
\multirow{2}{*}{AVE}       & Audio$\rightarrow$Visual   & \textbf{59.85} \\
                           & Visual$\rightarrow$Audio   & 56.33 \\ \midrule
\multirow{6}{*}{CMU-MOSEI} & Textual$\rightarrow$Audio$\rightarrow$Visual & \textbf{64.96} \\
                           & Text$\rightarrow$Visual$\rightarrow$Audio & 64.58 \\
                           & Audio$\rightarrow$Visual$\rightarrow$Text & 62.16 \\
                           & Audio$\rightarrow$Text$\rightarrow$Visual & 63.24 \\
                           & Visual$\rightarrow$Audio$\rightarrow$Text & 62.07 \\
                           & Visual$\rightarrow$Text$\rightarrow$Audio & 62.37 \\ \bottomrule[1pt]
\end{tabular}
\label{fig:order}
\end{wraptable}

\noindent \textbf{Impact of different modality training orders.} Table~\ref{fig:order} reports the performance of MCFT under different modality training orders on CREMA-D/AVE (two modalities) and CMU-MOSEI (three modalities). 
Empirically, on CREMA-D and AVE, the Audio → Visual schedule consistently outperforms Visual → Audio by a clear margin (58.36 vs. 56.09 on CREMA-D; 59.85 vs. 56.33 on AVE). 
On CMU-MOSEI, schedules that start with Text perform best, with Text → Audio → Visual achieving the top score (64.96), while starting from Audio/Visual leads to noticeably worse results (e.g., 62.07–63.24). We attribute this pattern to the design of MCFT: later modalities are trained with explicit cross-modal alignment and complementarity regularizers that depend on the set of preceding modalities, i.e., earlier stages effectively provide ``reference signals'' that shape subsequent optimization. Consequently, placing a more semantically reliable modality earlier (Audio in CREMA-D/AVE, Text in MOSEI) yields stronger guidance for later stages, whereas starting from a weaker/less stable modality can propagate suboptimal biases to downstream stages. 

\noindent \textbf{Sensitivity of hyperparameters ($\kappa, \pi$) in SSCA.}
Table~\ref{fig:kappapi} evaluates the hyperparameter sensitivity of SSCA by varying the retention ratio $\kappa$ and the consensus threshold $\pi$, two factors that directly control how much client update information is preserved and how conservatively cross-cluster updates are reconciled. 
Specifically, $\kappa$ retains the top-$\kappa$ fraction of coordinates (thus sparsifying $1-\kappa$), while $\pi$ gates coordinate-wise cross-cluster merging according to the directional dominance criterion. 

Regarding $\kappa$, performance improves markedly as $\kappa$ increases from $0.1$ to the range $0.7$--$0.8$, but then saturates or slightly degrades when $\kappa$ becomes very large (e.g. $0.9$). This behavior suggests a trade-off inherent to SSCA: overly small $\kappa$ corresponds to aggressive sparsification that discards many informative yet non-dominant coordinates, yielding an under-expressive global update and reducing the reliability of sign-based clustering. Conversely, overly large $\kappa$ retains almost the full update and therefore re-admits low-magnitude, high-variance coordinates that are most sensitive to client heterogeneity, which blurs the dominant directional structure exploited by SSCA and increases the chance that conflicting local drifts adversely affect the aggregated update.

For $\pi$, moderate-to-high values consistently outperform small thresholds, with the best results attained around $\pi \approx 0.8$. When $\pi$ is too small, SSCA merges coordinates even when directional dominance is weak, effectively over-smoothing across clusters and amplifying cross-cluster interference. When $\pi$ is too large (e.g., $0.9$), the criterion becomes overly conservative, suppressing merging even for coordinates with substantial but not overwhelming agreement, and thus leaving shared beneficial signal underutilized.

\noindent \textbf{Sensitivity of $\lambda_a$ and $\lambda_c$ on the three datasets.}
Figure~\ref{fig: lambda-4}(a--c) shows the effect of $\lambda_a$ and $\lambda_c$ on CREMA-D, AVE, and CMU-MOSEI. Overall across the three datasets, the strongest performance is achieved when $\lambda_a$ is in the range $0.3$--$0.5$ and $\lambda_c$ is in the range $0.8$--$1.2$ (with our final choice $\lambda_a=0.4$ and $\lambda_c=1.0$ lying near the optimum). This observation suggests that moderate $\lambda_a$ provides sufficient cross-modal alignment to transfer stable semantics from preceding modalities, while moderate $\lambda_c$ strengthens complementary learning by emphasizing hard/uncertain samples; excessively small values weaken both effects, whereas overly large values over-regularize the stage objective and can amplify noisy reweighting signals, reducing generalization.

\noindent \textbf{Sensitivity of $\lambda_{\mathrm{merge}}$ on the three datasets.}
Figure~\ref{fig: lambda-4}(d) reports the effect of $\lambda_{\mathrm{merge}}$ on the three datasets. Overall, performance peaks when $\lambda_{\mathrm{merge}}$ is in the range $0.8$--$0.9$ (with our final choice $\lambda_{\mathrm{merge}}=0.9$). The reason is that $\lambda_{\mathrm{merge}}$ controls how strongly the SSCA-aggregated consensus update is injected into the global model: too small values make the server update overly conservative and slow to accumulate consensus, while too large values make the update overly aggressive and more sensitive to residual cross-client inconsistency, slightly hurting stability.

\subsection{Modality Imbalance Ratio (MIR)}

We use the modality imbalance ratio (MIR) to quantify the degree of modality competition, i.e., how unevenly the model performs when restricted to different modalities.
For each modality $m$, we evaluate the model in a unimodal inference manner by enabling only modality $m$ (and disabling/masking other modalities) on all clients (including multimodal clients), and obtain the unimodal accuracy for each client.
We then average these unimodal accuracies over all clients to get a modality-wise mean accuracy $\bar{\mathrm{ACC}}^{(m)}$.
Finally, MIR is defined as the ratio between the best and the worst modality-wise mean accuracies:
\begin{equation}
\mathrm{MIR}=\frac{\max_{m}\ \bar{\mathrm{ACC}}^{(m)}}{\min_{m}\ \bar{\mathrm{ACC}}^{(m)}}.
\end{equation}
By construction, $\mathrm{MIR}\ge 1$, and a larger MIR indicates a larger modality-induced performance imbalance.

\begin{table*}[h]
\caption{The effect of hyperparameters $\kappa$ and $\pi$ on the performance of the three datasets.}
\centering
\begin{tabular}{l|ccccccccc}
\toprule
\multicolumn{1}{c|}{$\kappa$}  & $\kappa=0.1$  & $\kappa=0.2$  & $\kappa=0.3$ & $\kappa=0.4$  & $\kappa=0.5$  & $\kappa=0.6$  & $\kappa=0.7$  & $\kappa=0.8$  & $\kappa=0.9$  \\ \midrule
CREMA-D                 & 50.16 & 51.84 & 53.12 & 55.25 & 56.34 & 58.01 & \textbf{58.36} & 58.14 & 57.79 \\
AVE                     & 54.38 & 54.97 & 55.43 & 57.31 & 58.66 & 59.18 & \textbf{59.85} & 59.73 & 59.44 \\
CMU-MOSEI               & 60.67 & 61.68 & 62.24 & 62.53 & 63.75 & 64.27 & 64.84 & \textbf{64.96} & 64.65 \\ \midrule[0.8pt]
\multicolumn{1}{c|}{$\pi$} & $\pi=0.1$  & $\pi=0.2$  & $\pi=0.3$ & $\pi=0.4$  & $\pi=0.5$  & $\pi=0.6$  & $\pi=0.7$  & $\pi=0.8$  & $\pi=0.9$  \\ \midrule
CREMA-D                 & 56.43 & 56.29 & 56.82 & 57.02 & 57.33 & 58.39 & 58.21 & \textbf{58.36} & 58.14 \\
AVE                     & 55.92 & 56.47 & 57.55 & 57.32 & 57.86 & 58.48 & 59.24 & \textbf{59.85} & 59.45 \\
CMU-MOSEI               & 59.36 & 59.73 & 60.31 & 60.56 & 61.43 & 63.76 & 63.99 & \textbf{64.96} & 64.76 \\ \bottomrule
\end{tabular}\label{fig:kappapi}
\end{table*}

\begin{figure*}[h]
\centering
\includegraphics[width=1\textwidth]{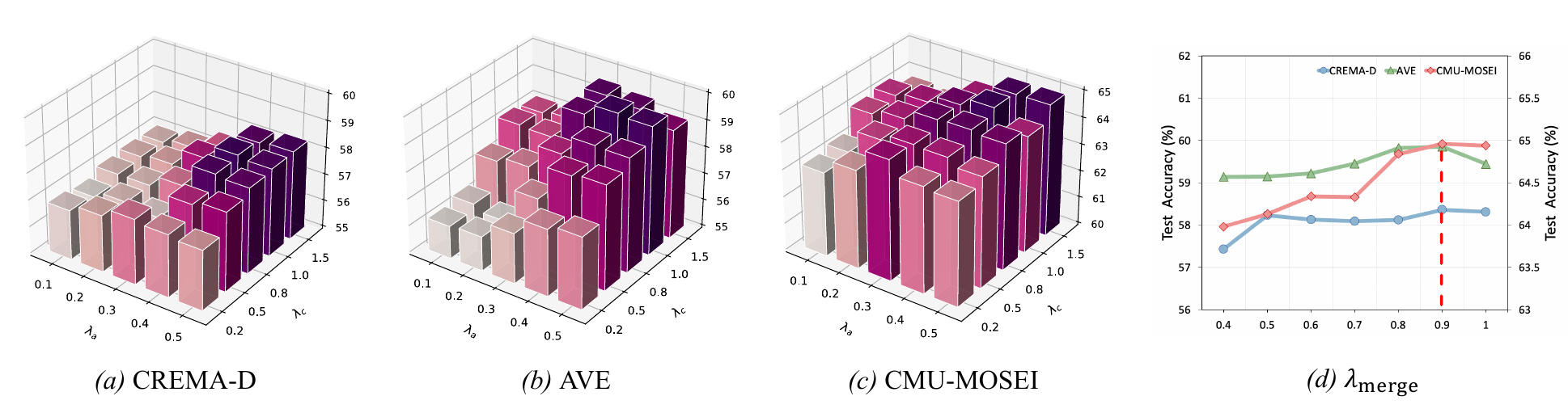}
\caption{(a)-(c): The effect of $\lambda_a$ and $\lambda_c$ on the performance of the three datasets. (d): The effect of $\lambda_{\text{merge}}$ on the performance of the three datasets.}
\label{fig: lambda-4}
\end{figure*}


\newpage
\section{Convergence Analysis}
\label{app:convergence}

\subsection{Problem Setup}
\paragraph{Modality-stage view.}
MCFT proceeds in modality-specific stages by activating one modality at a time.
In a fixed stage $m$, only the parameter block associated with the active modality is updated, while all other modality branches are frozen (see main text).
Accordingly, the analysis below focuses on the stage-$m$ parameter block, denoted by $\Theta \in \mathbb{R}^{d_m}$.

\paragraph{Cluster-specific objective.}
SSCA maintains $K$ cluster-specific global models whose parameter blocks evolve independently.
Fix a cluster index $k \in \{1,\ldots,K\}$ and let $\mathcal{C}_k$ denote the set of clients assigned to the cluster model $k$ in the current stage.
We define the stage-$m$, cluster-$k$ global objective as
\begin{equation}
F^{(m,k)}(\Theta) := \sum_{i\in \mathcal{C}_k} p_i \, \mathbb{E}_{\xi\sim \mathcal{D}_i}\!\left[f^{(m)}_i(\Theta;\xi)\right],
\quad p_i \ge 0,\ \sum_{i\in \mathcal{C}_k} p_i = 1,
\end{equation}
where $f^{(m)}_i$ is the stage-$m$ local loss used by client $i$ (Task + Alignment + Complementarity, as defined in the main text).
For brevity, we write $F(\Theta)$ when $(m,k)$ is clear from context.





\paragraph{Local update and SSCA aggregation.}
Let $r = 0,1,2,\ldots$ index the server aggregation rounds for the cluster model $k$ in stage $m$.
At the beginning of round $r$, the server broadcasts $\Theta_r \equiv \Theta^{(m,k)}_r$ to clients in $\mathcal{C}_k$.
Each client $i \in \mathcal{C}_k$ performs $\mathcal{E}$ local SGD steps:
\begin{equation}
\label{eq:app_local_sgd}
\Theta^{\rm{s}+1}_{i,r} = \Theta^{\rm{s}}_{i,r} - \eta\, g^{(m)}_i(\Theta^{\rm{s}}_{i,r}; \xi^{\rm{s}}_{i,r}),
\  \rm{s} = 0,1,\ldots,\mathcal{E}-1,
\  \Theta^0_{i,r} = \Theta_r,
\end{equation}
where $g^{(m)}_i(\cdot;\xi)$ is a stochastic gradient of $f^{(m)}_i$ at sample $\xi$ and $\eta > 0$ is the local stepsize.

Define the client vector:
\begin{equation}
\Delta\Theta_{i,r} := \Theta^\mathcal{E}_{i,r} - \Theta_r.
\end{equation}
SSCA takes $\{\Delta\Theta_{i,r}\}_{i\in\mathcal{C}_k}$ as input and outputs an aggregated update direction for cluster $k$:
\begin{equation}
\Delta\widehat{\Theta}_r \equiv \Delta\widehat{\Theta}^{(m,k)}_r \in \mathbb{R}^{d_m}.
\end{equation}
The server updates the cluster-$k$ model with merge strength $\lambda_{\mathrm{merge}} > 0$:
\begin{equation}
\label{eq:app_server_update}
\Theta_{r+1} = \Theta_r + \lambda_{\mathrm{merge}}\,\eta_g\Delta\widehat{\Theta}_r. 
\end{equation}

\paragraph{Notation.}
The server update in the main text uses a single coefficient $\lambda_{\mathrm{merge}}$(main).
For analysis, we rewrite it as $\lambda_{\mathrm{merge}}\text{(main)}=\eta_{\mathrm{g}}\lambda_{\mathrm{merge}}$,
where $\eta_{\mathrm{g}}$ is a global stepsize and $\lambda_{\mathrm{merge}}$ denotes the merge strength.
This is a purely notational re-parameterization and does not change the algorithm.


\subsection{Assumptions}

\begin{assumption}[Smoothness]
\label{ass:smooth}
For any modality stage $m$ and any client $i$, the local objective $f_i^{(m)}(\Theta)$ is $L$-smooth: for all $\Theta,\Theta'$,
$\|\nabla f_i^{(m)}(\Theta)-\nabla f_i^{(m)}(\Theta')\|\le L\|\Theta-\Theta'\|$.
Consequently, the corresponding cluster objective $F^{(m,k)}(\Theta)$ is also $L$-smooth.
\end{assumption}

\begin{assumption}[Unbiased stochastic gradients and bounded variance]
\label{ass:stoch}
For all $i\in\mathcal{C}_k$ and all $\Theta$, the stochastic gradient satisfies
$\mathbb{E}[g_i^{(m)}(\Theta;\xi)]=\nabla f_i^{(m)}(\Theta)$ and
$\mathbb{E}\|g_i^{(m)}(\Theta;\xi)-\nabla f_i^{(m)}(\Theta)\|^2 \le \sigma^2$.
\end{assumption}

\begin{assumption}[Client heterogeneity bound]
\label{ass:hetero}
There exists $\zeta\ge 0$ such that for all $\Theta$,
\begin{equation}
\sum_{i\in\mathcal{C}_k} p_i\, \big\|\nabla f_i^{(m)}(\Theta)-\nabla F(\Theta)\big\|^2 \le \zeta^2.
\label{eq:app_hetero}
\end{equation}
\end{assumption}

\begin{assumption}[Expected descent correlation of SSCA update]
\label{ass:dir_align}
There exists a constant $\gamma>0$ such that for all rounds $r$,
\begin{equation}
\Big\langle \nabla F(\Theta_r),~ \mathbb{E}\big[\Delta\widehat{\Theta}_{r}\mid \Theta_r\big] \Big\rangle
~\le~
-\gamma\, \big\|\nabla F(\Theta_r)\big\|^2.
\label{eq:app_dir_align}
\end{equation}
\end{assumption}

\begin{assumption}[Bounded second moment of SSCA update]
\label{ass:second_moment}
Assume that $\mathbb{E}\!\left[\|\Delta\widehat\Theta_r\|^2 \mid \Theta_r\right]$ exists and is finite for all rounds $r$.
There exist constants $B> 0$ and $\sigma_{\mathrm{agg}}^2\ge 0$ such that for all rounds $r$,
\begin{equation}
\mathbb{E}\big[\|\Delta\widehat{\Theta}_{r}\|^2 \mid \Theta_r\big]
~\le~
B\,\|\nabla F(\Theta_r)\|^2 + \sigma_{\mathrm{agg}}^2.
\label{eq:app_second_moment}
\end{equation}
\end{assumption}

\begin{remark}[Interpretation of Assumptions~\ref{ass:dir_align}--\ref{ass:second_moment}]
Assumption~\ref{ass:dir_align} requires that the SSCA update is descent-correlated in expectation, 
i.e., $\langle \nabla F(\Theta_r), \mathbb{E}[\Delta\widehat{\Theta}_r\mid \Theta_r]\rangle < 0$.
Assumption~\ref{ass:second_moment} upper-bounds the conditional second moment of the SSCA update and permits an ``error floor'' $\sigma_{\mathrm{agg}}^2$ that captures additional noise introduced by Top-$\kappa$ sparsification, sign quantization, and coordinate-wise masking/merging.
\end{remark}

\subsection{Local descent for the stage-$m$ objective}
\begin{lemma}[One-step expected descent for local SGD]
\label{lem:local_one_step}
Suppose each $f_i^{(m)}$ is $L$-smooth and Assumption~\ref{ass:stoch} holds. Then for a local update
$\Theta^+ = \Theta - \eta\, g_i^{(m)}(\Theta;\xi)$ with $\eta\le \frac{1}{L}$,
\begin{equation}
\mathbb{E}\big[f_i^{(m)}(\Theta^+)\big]
~\le~
f_i^{(m)}(\Theta)
-\frac{\eta}{2}\big\|\nabla f_i^{(m)}(\Theta)\big\|^2
+\frac{L\eta^2}{2}\sigma^2.
\label{eq:app_local_one_step}
\end{equation}
\end{lemma}

\begin{proof}
By $L$-smoothness of $f_i^{(m)}$,
$f(\Theta^+)\le f(\Theta)+\langle \nabla f(\Theta),\Theta^+-\Theta\rangle+\frac{L}{2}\|\Theta^+-\Theta\|^2$.
Substitute $\Theta^+-\Theta=-\eta g$ and take expectation.
Use unbiasedness and $\mathbb{E}\|g\|^2\le \|\nabla f(\Theta)\|^2+\sigma^2$.
Finally, $0 < \eta\le 1/L$ implies $1-\frac{L\eta}{2}\ge \frac{1}{2}$.
\end{proof}

\begin{lemma}[$\mathcal{E}$-step local progress bound]
\label{lem:local_tau_step}
Under the conditions of Lemma~\ref{lem:local_one_step}, for any client $i$ and any aggregation round $r$,
\begin{equation}
\sum_{s=0}^{\mathcal{E}-1}\mathbb{E}\big\|\nabla f_i^{(m)}(\Theta_{i,r}^{s})\big\|^2
~\le~
\frac{2}{\eta}\Big(\mathbb{E} f_i^{(m)}(\Theta_{r})-\mathbb{E} f_i^{(m)}(\Theta_{i,r}^{\mathcal{E}})\Big)
+ L\eta\mathcal{E}\sigma^2.
\label{eq:app_local_tau_step}
\end{equation}
\end{lemma}

\begin{proof}
Apply Lemma~\ref{lem:local_one_step} to each $s=0,\dots,\mathcal{E}-1$ and sum. Rearrange terms.
\end{proof}

\subsection{Local drift under periodic aggregation}
Define the ``round-start ideal'' local update:
\begin{equation}
\Delta\Theta_{i,r}^{\mathrm{ideal}}
~:=~
-\eta\mathcal{E}\, \nabla f_i^{(m)}(\Theta_r),
\label{eq:app_ideal_local}
\end{equation}
and the local drift error:
\begin{equation}
e_{i,r}^{\mathrm{drift}}
~:=~
\Delta\Theta_{i,r}-\Delta\Theta_{i,r}^{\mathrm{ideal}}.
\label{eq:app_drift_def}
\end{equation}

\begin{lemma}[Local drift bound (second moment)]
\label{lem:local_drift}
Under Assumptions~\ref{ass:stoch}--\ref{ass:hetero} and the $L$-smoothness of $f_i^{(m)}$, there exist universal constants $c_1,c_2>0$
(depending only on $L$) such that for $\eta\le \frac{1}{L}$,
\begin{equation}
\sum_{i\in\mathcal{C}_k} p_i\, \mathbb{E}\big[\|e_{i,r}^{\mathrm{drift}}\|^2\big]
~\le~
c_1\,\eta^2\,\mathcal{E}(\mathcal{E}-1)\,\sigma^2
~+~
c_2\,\eta^2\,\mathcal{E}^2\,\zeta^2.
\label{eq:app_drift_bound}
\end{equation}
\end{lemma}
We follow the standard drift analysis for local SGD and periodic model averaging~\cite{stich2019local, yu2019parallel} and adapt it to the stage-$m$ objective.

\begin{proof}[Proof sketch]
Expanding the local update gives
$\Delta\Theta_{i,r} = -\eta\sum_{s=0}^{\mathcal{E}-1} g_i^{(m)}(\Theta_{i,r}^{s};\xi_{i,r}^{s})$.
For each $s$, add and subtract $\nabla f_i^{(m)}(\Theta_{i,r}^{s})$ to obtain
\[
g_i^{(m)}(\Theta_{i,r}^{s};\xi_{i,r}^{s})-\nabla f_i^{(m)}(\Theta_r)
=
\underbrace{g_i^{(m)}(\Theta_{i,r}^{s};\xi_{i,r}^{s})-\nabla f_i^{(m)}(\Theta_{i,r}^{s})}_{\text{stochastic noise}}
+
\underbrace{\nabla f_i^{(m)}(\Theta_{i,r}^{s})-\nabla f_i^{(m)}(\Theta_r)}_{\text{drift}}.
\]
The noise term forms a martingale difference sequence and, under Assumption~\ref{ass:stoch}, contributes a mean-squared accumulation on the order of $\eta^2\mathcal{E}\sigma^2$.
For the drift term, $L$-smoothness yields
$\|\nabla f_i^{(m)}(\Theta_{i,r}^{s})-\nabla f_i^{(m)}(\Theta_r)\|\le L\|\Theta_{i,r}^{s}-\Theta_r\|$.
Moreover, $\mathbb{E}\|\Theta_{i,r}^{s}-\Theta_r\|^2$ is bounded by recursively unrolling \eqref{eq:app_local_sgd} and invoking Assumptions~\ref{ass:stoch}--\ref{ass:hetero}.
Summing over $s=0,\ldots,\mathcal{E}-1$ gives the claimed scaling
$O(\eta^2\mathcal{E}(\mathcal{E}-1)\sigma^2+\eta^2\mathcal{E}^2\zeta^2)$.
\end{proof}


\subsection{SSCA as an Inexact Descent Direction}
\paragraph{From local deltas to SSCA direction.}
SSCA maps the collection of local deltas $\{\Delta\Theta_{i,r}\}_{i\in\mathcal{C}_k}$ to $\Delta\widehat{\Theta}_{r}$.
We model SSCA as an aggregation operator that returns a descent-correlated update with a bounded conditional second moment, formalized in Assumptions~\ref{ass:dir_align} and~\ref{ass:second_moment}.

Optionally, one may decompose $\sigma_{\mathrm{agg}}^2$ into contributions from (i) Top-$\kappa$ sparsification/sign quantization and (ii) coordinate-wise dominance masking, but the analysis below only requires the combined bound in \eqref{eq:app_second_moment}.

\begin{assumption}[SSCA stability w.r.t.\ client deltas]
\label{ass:ssca_stability}
There exists a constant $C_{\mathrm{ssca}}\ge 1$ such that for any two collections of client deltas
$\{u_i\}_{i\in\mathcal{C}_k}$ and $\{v_i\}_{i\in\mathcal{C}_k}$,
the SSCA mapping satisfies
\begin{equation}
\big\|\mathrm{SSCA}(\{u_i\}) - \mathrm{SSCA}(\{v_i\})\big\|^2
~\le~
C_{\mathrm{ssca}} \sum_{i\in\mathcal{C}_k} p_i \|u_i - v_i\|^2.
\label{eq:app_ssca_stability}
\end{equation}
\end{assumption}

\begin{remark}[On Assumption~\ref{ass:ssca_stability}]
Assumption~\ref{ass:ssca_stability} is a technical abstraction that controls how perturbations in client deltas
(e.g., those induced by multi-step local training) are amplified through the SSCA mapping.
While SSCA involves non-smooth operations such as Top-$\kappa$ selection, sign quantization, and discrete clustering,
\eqref{eq:app_ssca_stability} can be viewed as a piecewise stability condition that holds within regions where the selected support and cluster assignment are unchanged,
or as an expected stability condition with respect to the inherent randomness in the selection/clustering procedures.
\end{remark}

\paragraph{Ideal SSCA direction (no local drift).}
Using the ideal client deltas $\{\Delta\Theta_{i,r}^{\mathrm{ideal}}\}_{i\in\mathcal{C}_k}$ defined in \eqref{eq:app_ideal_local},
define the corresponding ideal SSCA direction:
\begin{equation}
\Delta\widehat{\Theta}_{r}^{\,\mathrm{ideal}}
~:=~
\mathrm{SSCA}\big(\{\Delta\Theta_{i,r}^{\mathrm{ideal}}\}_{i\in\mathcal{C}_k}\big).
\label{eq:app_ssca_ideal_output}
\end{equation}
The actual SSCA direction is $\Delta\widehat{\Theta}_{r}=\mathrm{SSCA}(\{\Delta\Theta_{i,r}\})$.

\begin{lemma}[Propagation of local drift through SSCA]
\label{lem:drift_to_ssca}
Under Assumption~\ref{ass:ssca_stability},
\begin{equation}
\mathbb{E}\big[\|\Delta\widehat{\Theta}_{r}-\Delta\widehat{\Theta}_{r}^{\,\mathrm{ideal}}\|^2\big]
~\le~
C_{\mathrm{ssca}} \sum_{i\in\mathcal{C}_k} p_i\,\mathbb{E}\big[\|e_{i,r}^{\mathrm{drift}}\|^2\big],
\label{eq:app_drift_to_ssca}
\end{equation}
where $e_{i,r}^{\mathrm{drift}}$ is defined in \eqref{eq:app_drift_def}. Consequently, by Lemma~\ref{lem:local_drift}, for $\eta\le \frac{1}{L}$,
\begin{equation}
\mathbb{E}\big[\|\Delta\widehat{\Theta}_{r}-\Delta\widehat{\Theta}_{r}^{\,\mathrm{ideal}}\|^2\big]
~\le~
C_{\mathrm{ssca}}\Big(
c_1\,\eta^2\,\mathcal{E}(\mathcal{E}-1)\,\sigma^2
+
c_2\,\eta^2\,\mathcal{E}^2\,\zeta^2
\Big).
\label{eq:app_drift_to_ssca_explicit}
\end{equation}
\end{lemma}

\begin{proof}
Apply Assumption~\ref{ass:ssca_stability} to $u_i=\Delta\Theta_{i,r}$ and $v_i=\Delta\Theta_{i,r}^{\mathrm{ideal}}$:
\[
\|\Delta\widehat{\Theta}_{r}-\Delta\widehat{\Theta}_{r}^{\,\mathrm{ideal}}\|^2
\le
C_{\mathrm{ssca}} \sum_{i\in\mathcal{C}_k} p_i \|\Delta\Theta_{i,r}-\Delta\Theta_{i,r}^{\mathrm{ideal}}\|^2
=
C_{\mathrm{ssca}} \sum_{i\in\mathcal{C}_k} p_i \|e_{i,r}^{\mathrm{drift}}\|^2.
\]
Taking conditional expectation given $\Theta_r$ gives \eqref{eq:app_drift_to_ssca}, and \eqref{eq:app_drift_to_ssca_explicit} follows by Lemma~\ref{lem:local_drift}.
\end{proof}

\subsection{Main result: Nonconvex convergence to a neighborhood}
\begin{theorem}[Stage-$m$, cluster-$k$ nonconvex convergence of FedMediator]
\label{thm:main_nonconvex}
Suppose Assumptions~\ref{ass:smooth}--\ref{ass:hetero} hold.
Assume further that Assumptions~\ref{ass:dir_align}--\ref{ass:second_moment} hold for the \emph{actual} SSCA direction
$\Delta\widehat{\Theta}_{r}$ produced by SSCA at round $r$.
Choose the \emph{server} stepsize $\eta_{\mathrm{g}}>0$ such that
\begin{equation}
\gamma - \frac{L\,\eta_{\mathrm{g}}\,\lambda_{\mathrm{merge}}\,B}{2} ~>~ 0
\qquad \text{(e.g., } \eta_{\mathrm{g}} < \frac{2\gamma}{L\lambda_{\mathrm{merge}}B}\text{ when }B>0\text{)}.
\label{eq:app_stepsize_cond}
\end{equation}
Then for any number of aggregation rounds $R\ge 1$,
\begin{equation}
\frac{1}{R}\sum_{r=0}^{R-1}\mathbb{E}\big\|\nabla F(\Theta_r)\big\|^2
~\le~
\frac{2\big(F(\Theta_0)-F^*\big)}{\eta_{\mathrm{g}}\,\lambda_{\mathrm{merge}}\,
\Big(2\gamma - L\eta_{\mathrm{g}}\lambda_{\mathrm{merge}}B\Big)\,R}
~+~
\frac{L\,\eta_{\mathrm{g}}\,\lambda_{\mathrm{merge}}\,\sigma_{\mathrm{agg}}^2}{\Big(2\gamma - L\eta_{\mathrm{g}}\lambda_{\mathrm{merge}}B\Big)},
\label{eq:app_rate_basic}
\end{equation}
where $F^*:=\inf_{\Theta}F(\Theta)$.
\end{theorem}

\begin{proof}
By $L$-smoothness (Assumption~\ref{ass:smooth}) and the update \eqref{eq:app_server_update},
\begin{align}
F(\Theta_{r+1})
&\le
F(\Theta_r)
+
\big\langle \nabla F(\Theta_r),~ \Theta_{r+1}-\Theta_r \big\rangle
+\frac{L}{2}\big\|\Theta_{r+1}-\Theta_r\big\|^2 \nonumber\\
&=
F(\Theta_r)
+
\eta_{\mathrm{g}}\lambda_{\mathrm{merge}}
\big\langle \nabla F(\Theta_r),~ \Delta\widehat{\Theta}_{r} \big\rangle
+
\frac{L\eta_{\mathrm{g}}^2\lambda_{\mathrm{merge}}^2}{2}\big\|\Delta\widehat{\Theta}_{r}\big\|^2.
\label{eq:app_smooth_step}
\end{align}
Take conditional expectation w.r.t.\ $\Theta_r$ and apply Assumptions~\ref{ass:dir_align}--\ref{ass:second_moment}:
\begin{align}
\mathbb{E}\big[F(\Theta_{r+1}) \mid \Theta_r\big]
&\le
F(\Theta_r)
-
\eta_{\mathrm{g}}\lambda_{\mathrm{merge}}\gamma\|\nabla F(\Theta_r)\|^2
+\frac{L\eta_{\mathrm{g}}^2\lambda_{\mathrm{merge}}^2}{2}\Big(B\|\nabla F(\Theta_r)\|^2+\sigma_{\mathrm{agg}}^2\Big)
\nonumber\\
&=
F(\Theta_r)
-
\eta_{\mathrm{g}}\lambda_{\mathrm{merge}}
\Big(\gamma-\frac{L\eta_{\mathrm{g}}\lambda_{\mathrm{merge}}B}{2}\Big)\,
\|\nabla F(\Theta_r)\|^2
+
\frac{L\eta_{\mathrm{g}}^2\lambda_{\mathrm{merge}}^2}{2}\sigma_{\mathrm{agg}}^2.
\label{eq:app_descent_ineq}
\end{align}
Now take total expectation and sum \eqref{eq:app_descent_ineq} over $r=0,\dots,R-1$:
\begin{equation}
\eta_{\mathrm{g}}\lambda_{\mathrm{merge}}
\Big(\gamma-\frac{L\eta_{\mathrm{g}}\lambda_{\mathrm{merge}}B}{2}\Big)
\sum_{r=0}^{R-1}\mathbb{E}\|\nabla F(\Theta_r)\|^2
\le
F(\Theta_0)-\mathbb{E}F(\Theta_R)
+
\frac{L\eta_{\mathrm{g}}^2\lambda_{\mathrm{merge}}^2}{2}R\sigma_{\mathrm{agg}}^2.
\label{eq:app_telescope}
\end{equation}
Using $F(\Theta_R)\ge F^*$ and dividing both sides by
$\eta_{\mathrm{g}}\lambda_{\mathrm{merge}}\big(\gamma-\frac{L\eta_{\mathrm{g}}\lambda_{\mathrm{merge}}B}{2}\big)R$ yields \eqref{eq:app_rate_basic}.
\end{proof}

\begin{remark}[Choosing stepsizes]
Condition \eqref{eq:app_stepsize_cond} ensures a \emph{positive descent coefficient}.
A sufficient choice is $\eta_{\mathrm{g}} = O\!\left(\frac{1}{\lambda_{\mathrm{merge}}}\right)$, and in particular
$\eta_{\mathrm{g}} < \frac{2\gamma}{L\lambda_{\mathrm{merge}}B}$ when $B>0$.
\end{remark}

\begin{remark}[From stage-wise to multi-stage (MCFT cycles)]
Theorem~\ref{thm:main_nonconvex} is stage-wise (fixed active modality $m$).
Under a cyclic schedule where each stage is activated regularly, one can sum the per-stage descent inequalities over a full cycle and obtain an analogous bound on the cycle-averaged gradient norm, with constants depending on the number of stages.
\end{remark}

\begin{remark}[A sufficient condition for Assumption~\ref{ass:second_moment} under periodic local training]
Theorem~\ref{thm:main_nonconvex} assumes Assumption~\ref{ass:second_moment} for the actual SSCA direction.
When periodic local training is used, one convenient way to verify this assumption is to start from the ideal SSCA direction
$\Delta\widehat{\Theta}_{r}^{\,\mathrm{ideal}}$ in \eqref{eq:app_ssca_ideal_output} and control the additional drift-induced error via
Assumption~\ref{ass:ssca_stability} and Lemma~\ref{lem:drift_to_ssca}.
In particular, if Assumptions~\ref{ass:dir_align}--\ref{ass:second_moment} hold for $\Delta\widehat{\Theta}_{r}^{\,\mathrm{ideal}}$ with
\begin{align}
\Big\langle \nabla F(\Theta_r),~ \mathbb{E}[\Delta\widehat{\Theta}_{r}^{\,\mathrm{ideal}}\mid \Theta_r]\Big\rangle
&\le -\gamma \|\nabla F(\Theta_r)\|^2, \label{eq:app_dir_align_ideal}\\
\mathbb{E}\big[\|\Delta\widehat{\Theta}_{r}^{\,\mathrm{ideal}}\|^2\mid \Theta_r\big]
&\le B\|\nabla F(\Theta_r)\|^2 + \sigma_{\mathrm{ssca}}^2, \label{eq:app_second_moment_ideal}
\end{align}
then the actual SSCA direction satisfies
\begin{equation}
\mathbb{E}\big[\|\Delta\widehat{\Theta}_{r}\|^2\mid \Theta_r\big]
~\le~
2B\|\nabla F(\Theta_r)\|^2
~+~
2\sigma_{\mathrm{ssca}}^2
~+~
2\,C_{\mathrm{ssca}}\Big(
c_1\,\eta^2\,\mathcal{E}(\mathcal{E}-1)\,\sigma^2
+
c_2\,\eta^2\,\mathcal{E}^2\,\zeta^2
\Big),
\label{eq:app_second_moment_with_drift}
\end{equation}
where $\eta$ is the \emph{local} SGD stepsize used in \eqref{eq:app_local_sgd}, and $\eta_{\mathrm{g}}$ is the \emph{server} stepsize used in \eqref{eq:app_server_update}.
Thus Assumption~\ref{ass:second_moment} holds for $\Delta\widehat{\Theta}_{r}$ with updated constants
$B \leftarrow 2B$ and
$\sigma_{\mathrm{agg}}^2 \leftarrow 2\sigma_{\mathrm{ssca}}^2 + 2\,C_{\mathrm{ssca}}\big(c_1\,\eta^2\,\mathcal{E}(\mathcal{E}-1)\,\sigma^2 + c_2\,\eta^2\,\mathcal{E}^2\,\zeta^2\big)$.
\end{remark}

\paragraph{Refinement: Separating SSCA Aggregation Noise and Drift.}
If desired, one can keep $\sigma_{\mathrm{ssca}}^2$ (pure aggregation noise) and the drift term separate in \eqref{eq:app_rate_basic},
by replacing $\sigma_{\mathrm{agg}}^2$ with $\sigma_{\mathrm{ssca}}^2$ and adding an extra term proportional to
$\eta_{\mathrm{g}}\lambda_{\mathrm{merge}}\big(c_1\,\eta^2\,\mathcal{E}(\mathcal{E}-1)\sigma^2+c_2\,\eta^2\,\mathcal{E}^2\zeta^2\big)$.
This yields the same qualitative conclusion: periodic aggregation ($\mathcal{E}>1$) increases the neighborhood radius through the drift term.

\newpage
\section{Pseudo-code}\label{app:code}
This appendix provides the pseudo-code of \textsc{FedMChain}, including the chain-of-modalities training schedule, the client-side \textsc{LocalUpdate}, and the server-side \textsc{SSCA} aggregation used in each modality stage.

\begin{algorithm}[H]
\caption{\textsc{FedMChain}}
\begin{algorithmic}[1]
\State \textbf{Input:} Clients $\mathcal{C}=\{1,\ldots,C\}$ with local datasets $\{D_i\}_{i\in\mathcal{C}}$; global modality set $\mathcal{M}$; modality-specific global parameters $\{\Theta^{(m)}_{G}\}_{m\in\mathcal{M}}$; local learning rate $\eta$; local steps $\mathcal{E}$; per-modality rounds $\{R_m\}_{m\in\mathcal{M}}$; SSCA hyper-parameters $\kappa$ (sparsification), $K$ (number of clusters), $\pi$ (threshold), $\lambda_{\mathrm{merge}}$ (merge rate); regularizer weights $\lambda_a,\lambda_c$ and temperature $\tau$.
\State \textbf{Output:} Trained $\{\Theta^{(m)}_{G}\}_{m\in\mathcal{M}}$.

\Statex
\State Initialize $\{\Theta^{(m)}_{G}\}_{m\in\mathcal{M}}$
\For{each modality $m \in \mathcal{M}$} \textcolor{blue}{\Comment{Modality-Chained schedule}}
    \For{$r = 1,2,\ldots,R_m$}
        \State Server samples participating clients $\mathcal{S}^{(m)}_r \subseteq \mathcal{C}$
        \State Broadcast $\Theta^{(m)}_{G}$ to all $i\in \mathcal{S}^{(m)}_r$
        \For{each client $i \in \mathcal{S}^{(m)}_r$ \textbf{in parallel}}
            \State $\Delta \Theta^{(m)}_{i} \gets \textsc{LocalUpdate}(i, m, \Theta^{(m)}_{G}, \eta, \mathcal{E}, \lambda_a,\lambda_c,\tau)$
            \State Send $\Delta \Theta^{(m)}_{i}$ and $w_i$ to server \textcolor{blue}{\Comment{$w_i$ as in Eq.~\ref{eq:10}}}
        \EndFor
        \State $\Delta \widehat{\Theta}^{(m)} \gets \textsc{SSCA}(\{\Delta \Theta^{(m)}_{i}, w_i\}_{i\in\mathcal{S}^{(m)}_r}; \kappa,K,\pi)$
        \State $\Theta^{(m)}_{G} \gets \Theta^{(m)}_{G} + \lambda_{\mathrm{merge}} \cdot \Delta \widehat{\Theta}^{(m)}$ 
    \EndFor
\EndFor

\Statex
\State \textsc{\textbf{LocalUpdate}}$(i, m, \Theta^{(m)}_{G}, \eta, \mathcal{E}, \lambda_a,\lambda_c,\tau)$:
    \State $\Theta^{(m)}_{i} \gets \Theta^{(m)}_{G}$
    \For{$e=1,2,\ldots,\mathcal{E}$}
        \State Sample mini-batch $B_i \subset D_i$
        \State Freeze all modality branches except modality $m$
        \State Update modality-$m$ parameters by minimizing the stage objective in Eq.~\ref{eq:4}--~\ref{eq:8}
        \State $\Theta^{(m)}_{i} \gets \Theta^{(m)}_{i} - \eta \nabla_{\Theta^{(m)}_{i}} \mathcal{L}^{(m)}(B_i)$
    \EndFor
    \State $\Delta \Theta^{(m)}_{i} \gets \Theta^{(m)}_{i} - \Theta^{(m)}_{G}$
    \State \Return $\Delta \Theta^{(m)}_{i}$

\Statex
\State \textsc{\textbf{SSCA}}$(\{\Delta \Theta^{(m)}_{i}, w_i\}_{i\in\mathcal{S}^{(m)}}; \kappa,K,\pi)$:
    \State Compute sparse sign vectors $s_i = \mathrm{sign}\!\big(T(\Delta \Theta^{(m)}_{i}, \kappa)\big)$ for all $i\in\mathcal{S}$
    \State Cluster $\{s_i\}$ into $\{\mathcal{C}_1,\ldots,\mathcal{C}_K\}$
    \State Compute cluster consensuses $\{\Delta \bar{\Theta}^{(m,k)}\}_{k=1}^{K}$ by weighted mean \textcolor{blue}{\Comment{Eq.~\ref{eq:10}}}
    \State Compute $\rho_q$, $\mu_q$, and $d_q$ coordinate-wisely \textcolor{blue}{\Comment{Eq.~\ref{eq:11}--~\ref{eq:14}}}
    \State Merge cluster consensuses into $\Delta \widehat{\Theta}^{(m)}$ using Eq.~\ref{eq:15}--~\ref{eq:16}
    \State \Return $\Delta \widehat{\Theta}^{(m)}$



\Statex
\State \Return $\{\Theta^{(m)}_{G}\}_{m\in\mathcal{M}}$
\end{algorithmic}
\end{algorithm}



\end{document}